\documentclass[12pt]{article}
\usepackage{amsmath}
\usepackage{amsfonts}
\usepackage{makeidx}
\usepackage{amssymb}
\usepackage{times}
\usepackage{anysize}

\marginsize{3cm}{3cm}{2cm}{2cm}
\newtheorem{satz}{Theorem}[section]
\newtheorem{lemma}[satz]{Lemma}
\newtheorem{koro}[satz]{Corollary}

\newtheorem{bemerkung}[satz]{Remark}

\newtheorem{proposition}[satz]{Proposition}
\newtheorem{notation}[satz]{Notation}
\newenvironment{proof}{\par\noindent {\it Proof:} \hspace{7pt}}{\hfill\hbox{\vrule width 7pt depth 0pt height 7pt}
\par\vspace{10pt}}

\input{tcilatex}
\begin{document}

\title{Macroscopic Conductivity of Free Fermions in Disordered Media}
\author{J.-B. Bru \and W. de Siqueira Pedra \and C. Hertling}
\date{\today}
\maketitle

\begin{abstract}
We conclude our analysis of the linear response of charge transport in
lattice systems of free fermions subjected to a random potential by deriving
general mathematical properties of its conductivity at the macroscopic
scale. The present paper belongs to a succession of studies on Ohm and
Joule's laws from a thermodynamic viewpoint starting with \cite%
{OhmI,OhmII,OhmIII}. We show, in particular, the existence and finiteness of
the conductivity measure $\mu _{\mathbf{\Sigma }}$ for macroscopic scales.
Then we prove that, similar to the conductivity measure associated to
Drude's model, $\mu _{\mathbf{\Sigma }}$ converges in the weak$^{\ast } $%
--topology to the trivial measure in the case of perfect insulators (strong
disorder, complete localization), whereas in the limit of perfect conductors
(absence of disorder) it converges to an atomic measure concentrated at
frequency $\nu =0$. However, the AC--conductivity $\mu _{\mathbf{\Sigma }}|_{%
\mathbb{R}\backslash \{0\}}$ does not vanish in general: We show that $\mu _{%
\mathbf{\Sigma }}(\mathbb{R}\backslash \{0\})>0$, at least for large
temperatures and a certain regime of small disorder. \bigskip
\end{abstract}

\noindent\textbf{Keywords:} disordered systems, transport processes,
conductivity measure, Anderson model \smallskip

\noindent \textbf{Mathematics Subject Classification 2010:} 82C70, 82C44,
82C20


\section{Introduction}

We define in \cite{OhmIII} AC--conductivity measures for free fermions on
the lattice subjected to a random potential by using the second principle of
thermodynamics, which corresponds to the positivity of the heat production
for cyclic processes on equilibrium states. Such measures were introduced
for the first time in \cite{Annale,JMP-autre} by using a different approach.

In \cite{OhmIII} we prove moreover Ohm and Joule's laws from first
principles of thermodynamics and quantum mechanics for electric fields that
is time-- and space--dependent. The microscopic theory usually explaining
these laws is based on Drude's model (1900) combined with quantum
corrections. [Cf. the Landau theory of fermi liquids.] Indeed, although the
motion of electrons and ions is treated classically and the interaction
between these two species is modeled by perfectly elastic random collisions,
this quite elementary model provides a qualitatively good description of
DC-- and AC--conductivities in metals. Recall that well--known computations
using Drude's model predict that the conductivity $\Sigma _{\mathrm{Drude}%
}(t)$ behaves like%
\begin{equation}
\Sigma _{\mathrm{Drude}}(t)=D\exp (-\mathrm{T}^{-1}t)\ ,\qquad t\in \mathbb{R%
}_{0}^{+}\ ,  \label{eq0}
\end{equation}%
where $\mathrm{T}>0$ is related to the mean time interval between two
collisions of a charged carrier with defects in the crystal, whereas $D\in
\mathbb{R}^{+}$ is some strictly positive constant. In particular, for any
electromagnetic potential $\mathbf{A}\in C_{0}^{\infty }(\mathbb{R}\times {%
\mathbb{R}}^{3};({\mathbb{R}}^{3})^{\ast })$ with corresponding electric
field (in the Weyl gauge)
\begin{equation*}
E_{\mathbf{A}}(t,x):=-\partial _{t}\mathbf{A}(t,x)\ ,\quad t\in \mathbb{R},\
x\in \mathbb{R}^{3}\ ,
\end{equation*}
the heat production at large times is in this case equal to%
\begin{equation*}
\int\nolimits_{t_{0}}^{t}\mathrm{d}s_{1}\int\nolimits_{t_{0}}^{s_{1}}\mathrm{%
d}s_{2}\Sigma _{\mathrm{Drude}}(s_{1}-s_{2})\int\nolimits_{\mathbb{R}^{3}}%
\mathrm{d}^{3}x\langle E_{\mathbf{A}}(s_{2},x),E_{\mathbf{A}}(s_{1},x)\rangle
\end{equation*}%
for any $t\geq t_{0}$, where $t_{0}$ is the time when the electromagnetic
potential is turned on, i.e., $\mathbf{A}(t,\cdot )=0$ for all $t\leq t_{0}$%
. Then, since $s\mapsto E_{\mathbf{A}}(s,x)$ is smooth and compactly
supported for all $x\in \mathbb{R}^{3}$, we deduce from Fubini's theorem and
(\ref{eq0}) that%
\begin{eqnarray*}
&&\int\nolimits_{t_{0}}^{t}\mathrm{d}s_{1}\int\nolimits_{t_{0}}^{s_{1}}%
\mathrm{d}s_{2}\Sigma _{\mathrm{Drude}}(s_{1}-s_{2})\int\nolimits_{\mathbb{R}%
^{3}}\mathrm{d}^{3}x\left\langle E_{\mathbf{A}}(s_{2},x),E_{\mathbf{A}%
}(s_{1},x)\right\rangle \\
&=&\frac{1}{2}\int\nolimits_{\mathbb{R}^{3}}\mathrm{d}^{3}x\int_{\mathbb{R}}%
\mathrm{d}\nu |\hat{E}_{\mathbf{A}}(\nu ,x)|^{2}\vartheta _{\mathrm{T}%
}\left( \nu \right) \ ,
\end{eqnarray*}%
where $\nu \mapsto \hat{E}_{\mathbf{A}}(\nu ,x)$ and%
\begin{equation*}
\nu \mapsto \vartheta _{\mathrm{T}}\left( \nu \right) \sim \frac{\mathrm{T}}{%
1+\mathrm{T}^{2}\nu ^{2}}
\end{equation*}%
are the Fourier transforms of the maps
\begin{equation*}
s\mapsto E_{\mathbf{A}}(s,x)\qquad \text{and}\qquad s\mapsto \exp \left( -%
\mathrm{T}^{-1}\left\vert s\right\vert \right) \ ,
\end{equation*}%
respectively, at any fixed $x\in \mathbb{R}^{3}$. In particular,
\begin{equation*}
|\hat{E}_{\mathbf{A}}(\nu ,x)|^{2}\vartheta _{\mathrm{T}}\left( \nu \right)
\mathrm{d}\nu
\end{equation*}%
is the heat production due to the component of frequency $\nu $ of the
electric field, in accordance with Joule's law in the AC--regime.

Thus, the (positive) measure $\vartheta _{\mathrm{T}}(\nu )\mathrm{d}\nu $
is the in--phase conductivity measure of Drude's model. Its restriction to $%
\mathbb{R}\backslash \{0\}$ can be interpreted as an (in--phase) \emph{AC}%
--conductivity measure. In the limit of the perfect insulator ($\mathrm{T}%
\rightarrow 0$) the in--phase conductivity measure $\vartheta _{\mathrm{T}%
}(\nu )\mathrm{d}\nu $ converges in the weak$^{\ast }$--topology to the
trivial measure ($0\cdot \mathrm{d}\nu $). On the other hand, in the limit
of the perfect conductor ($\mathrm{T}\rightarrow \infty $), only the
in--phase AC--conductivity measure of Drude's model, as defined above,
converges in the weak$^{\ast }$--topology to the trivial measure ($0\cdot
\mathrm{d}\nu $) on $\mathbb{R}\backslash \{0\}$. Indeed, as $\mathrm{T}%
\rightarrow \infty $, the in--phase conductivity measure $\vartheta _{%
\mathrm{T}}(\nu )\mathrm{d}\nu $ converges in the weak$^{\ast }$--topology
to the atomic measure $D\delta_0 $ concentrated at $\nu=0$ with $D\in
\mathbb{R}^{+} $ being some strictly positive constant. Here, $\delta_0 (B
):= \mathbf{1}[0 \in B]$ for any Borel set $B \subset \mathbb{R}$.

One aim of this paper is to verify this phenomenology for our many--body
quantum system. To this end, we represent the conductivity measure -- up to
some explicit atomic correction at zero frequency ($\nu=0$) -- as the
spectral measure of some self adjoint operator with respect to (w.r.t.) a
fixed vector. This proof uses analyticity properties of correlation
functions of KMS states. It involves the so--called Duhamel two--point
function as explained in \cite[Section A]{OhmII} and requires the
construction of a Hilbert space of (here called) ``current \emph{Duhamel}
fluctuations''. Using these objects we derive various mathematical
properties of the conductivity $\mathbf{\Sigma }$ of the fermion system. In
particular, $\mathbf{\Sigma }$ is shown to be a time--correlation function
of some unitary evolution. This yield the existence of the conductivity
measure $\mu _{\mathbf{\Sigma }}$ as a spectral measure (up to an explicit
atomic correction).

Another important outcome of this approach is the finiteness of $\mu _{%
\mathbf{\Sigma }}$, i.e., $\mu _{\mathbf{\Sigma }}(\mathbb{R})$ $<\infty $.
Moreover, the conductivity measure is not anymore restricted to $\mathbb{R}%
\backslash \{0\}$, in contrast with \cite{OhmIII}.

Similar to Drude's model, we also show that the AC--conductivity measure $%
\mu _{\mathbf{\Sigma }}|_{\mathbb{R}\backslash \{0\}}$ converges in the weak$%
^{\ast }$--topology to the trivial measure in the case of perfect
conductors, i.e., the absence of disorder, as well as in the case of perfect
insulators, i.e., in the case of strong disorder. Note that the fact that
the AC--conducti%
\-%
vity measure becomes zero does not imply, in general, that there is no
current in presence of electric fields. It only implies that the so--called
in--phase current, which is the component of the total current producing
heat, also called active current, is zero. Furthermore, the AC--conducti%
\-%
vity $\mu _{\mathbf{\Sigma }}|_{\mathbb{R}\backslash \{0\}}$ is in general
non--vanishing: We show in Theorem \ref{main 2 copy(1)} that $\mu _{\mathbf{%
\Sigma }}(\mathbb{R}\backslash \{0\})>0$, for large temperatures and a
certain regime of small disorder.

More precisely, we show that, for any cyclic process driven by the external
electric field, the heat production vanishes in both limits of perfect
conductors and perfect insulators, but the full conductivity does not vanish
in the case of perfect conductors (cf. Theorem \ref{main 4}). In this last
case, \emph{exactly} like in Drude's model, the conductivity measure $\mu _{%
\mathbf{\Sigma }}$ converges in the weak$^{\ast }$--topology to the atomic
measure $\tilde{D}\delta_0 $ with $\tilde{D}\in \mathbb{R}^{+}$ being the
\emph{explicit} strictly positive constant (\ref{explicit constant}) and $%
\delta_0 (B) := \mathbf{1}[0 \in B]$ for any Borel set $B \subset \mathbb{R}$%
.


To conclude, our main assertions are Theorems \ref{toto fluctbis+D} (current
Duhamel fluctuations), \ref{lemma sigma pos type copy(4)-macro}
(mathematical properties of the paramagnetic conductivity), \ref{main 4}
(asymptotic behavior of the conductivity), and \ref{main 2 copy(1)} (strict
positivity of the heat production). This paper is organized as follows:

\begin{itemize}
\item The random fermion system is defined in Section \ref{Section main
results}. The mathematical framework of this study is the one of \cite%
{OhmI,OhmII,OhmIII}.

\item In Section \ref{section Current Fluctuations copy(1)} we define the
Hilbert space of ``current Duhamel fluctuations''.

\item In Section \ref{Sect Conductivity of Fermion Systems} we derive
important mathematical properties of the conductivity of the fermion system.

\item Section \ref{Sect tehnical conduc} gathers technical proofs related to
the asymptotic behavior of the conductivity and the strict positivity of the
heat production. Both studies use explicit computations based on results of
\cite{OhmII,OhmIII}.
\end{itemize}

\begin{notation}[Generic constants]
\label{remark constant}\mbox{
}\newline
To simplify notation, we denote by $D$ any generic positive and finite
constant. These constants do not need to be the same from one statement to
another.
\end{notation}

\section{Setup of the Problem\label{Section main results}}

Let $d\in \mathbb{N}$, $\mathfrak{L}:=\mathbb{Z}^{d}$ and $(\Omega ,%
\mathfrak{A}_{\Omega },\mathfrak{a}_{\Omega })$ be the probability space
defined as follows: Set $\Omega :=[-1,1]^{\mathfrak{L}}$ and let $\Omega
_{x} $, $x\in \mathfrak{L}$, be an arbitrary element of the Borel $\sigma $%
--algebra of the interval $[-1,1]$ w.r.t. the usual metric topology. Then, $%
\mathfrak{A}_{\Omega }$ is the $\sigma $--algebra generated by cylinder sets
$\prod\nolimits_{x\in \mathfrak{L}}\Omega _{x}$, where $\Omega _{x}=[-1,1]$
for all but finitely many $x\in \mathfrak{L}$. The measure $\mathfrak{a}%
_{\Omega }$ is the product measure
\begin{equation}
\mathfrak{a}_{\Omega }%
\Big(%
\underset{x\in \mathfrak{L}}{\prod }\Omega _{x}%
\Big)%
:=\underset{x\in \mathfrak{L}}{\prod }\mathfrak{a}_{\mathbf{0}}(\Omega
_{x})\ ,  \label{probability measure}
\end{equation}%
where $\mathfrak{a}_{\mathbf{0}}$ is any fixed probability measure on the
interval $[-1,1]$. We denote by $\mathbb{E}[\ \cdot \ ]$ the expectation
value associated with $\mathfrak{a}_{\Omega }$.

For simplicity and without loss of generality (w.l.o.g.), we assume that the
expectation of the random variable at any single site is zero:
\begin{equation}
\mathbb{E}\left[ \omega (0)\right] =\int_{\Omega }\omega (0)\mathrm{d}%
\mathfrak{a}_{\mathbf{0}}(\omega )=0\ .  \label{V expectation}
\end{equation}%
We can easily remove this condition by replacing $\omega $ by $\omega -%
\mathbb{E}[\omega (0)]$ and adding $\mathbb{E}[\omega (0)]$ to the discrete
Laplacian defined below.

Note that the i.i.d. property of the potential is not essential for our
results. We could take any ergodic ensemble instead. However, this
assumption and (\ref{V expectation}) extremely simplify the proof of the
asymptotic behavior of the conductivity (Theorem \ref{main 4}) and of the
strict positivity of the heat production (Theorem \ref{main 2 copy(1)}).

For any realization $\omega \in \Omega $, $V_{\omega }\in \mathcal{B}(\ell
^{2}(\mathfrak{L}))$ is the self--adjoint multiplication operator with the
function $\omega :\mathfrak{L}\rightarrow \lbrack -1,1]$. Then we consider
the \emph{Anderson tight--binding model} $(\Delta _{\mathrm{d}}+\lambda
V_{\omega })$ acting on the Hilbert space $\ell ^{2}(\mathfrak{L})$, where $%
\Delta _{\mathrm{d}}\in \mathcal{B}(\ell ^{2}(\mathfrak{L}))$ is (up to a
minus sign) the usual $d$--dimensional discrete Laplacian given by%
\begin{equation}
\lbrack \Delta _{\mathrm{d}}(\psi )](x):=2d\psi (x)-\sum\limits_{z\in
\mathfrak{L},\text{ }|z|=1}\psi (x+z)\ ,\text{\qquad }x\in \mathfrak{L},\
\psi \in \ell ^{2}(\mathfrak{L})\ .  \label{discrete laplacian}
\end{equation}%
To define the one--particle dynamics like in \cite{Annale}, we use the
unitary group $\{\mathrm{U}_{t}^{(\omega ,\lambda )}\}_{t\in \mathbb{R}}$
generated by the random Hamiltonian $(\Delta _{\mathrm{d}}+\lambda V_{\omega
})$ for $\omega \in \Omega $ and $\lambda \in \mathbb{R}_{0}^{+}$:%
\begin{equation}
\mathrm{U}_{t}^{(\omega ,\lambda )}:=\exp (-it(\Delta _{\mathrm{d}}+\lambda
V_{\omega }))\in \mathcal{B}(\ell ^{2}(\mathfrak{L}))\ ,\text{\qquad }t\in
\mathbb{R}\ .  \label{rescaled}
\end{equation}

Denote by $\mathcal{U}$ the CAR $C^{\ast }$--algebra associated to the
infinite system. Annihilation and creation operators of (spinless) fermions
with wave functions $\psi \in \ell ^{2}(\mathfrak{L})$ are defined by
\begin{equation*}
a(\psi ):=\sum\limits_{x\in \mathfrak{L}}\overline{\psi (x)}a_{x}\in
\mathcal{U}\ ,\quad a^{\ast }(\psi ):=\sum\limits_{x\in \mathfrak{L}}\psi
(x)a_{x}^{\ast }\in \mathcal{U}\ .
\end{equation*}%
Here, $\{a_{x},a_{x}^{\ast }\}_{x\in \mathfrak{L}}\subset \mathcal{U}$ and
the identity $\mathbf{1}\in \mathcal{U}$ are generators of $\mathcal{U}$ and
satisfy the canonical anti--commutation relations. For all $\omega \in
\Omega $ and $\lambda \in \mathbb{R}_{0}^{+}$, the condition%
\begin{equation}
\tau _{t}^{(\omega ,\lambda )}(a(\psi ))=a((\mathrm{U}_{t}^{(\omega ,\lambda
)})^{\ast }(\psi ))\ ,\text{\qquad }t\in \mathbb{R},\ \psi \in \ell ^{2}(%
\mathfrak{L})\ ,  \label{rescaledbis}
\end{equation}%
uniquely defines a family $\tau ^{(\omega ,\lambda )}:=\{\tau _{t}^{(\omega
,\lambda )}\}_{t\in {\mathbb{R}}}$ of (Bogoliubov) automorphisms of $%
\mathcal{U}$, see \cite[Theorem 5.2.5]{BratteliRobinson}. The one--parameter
group $\tau ^{(\omega ,\lambda )}$ is strongly continuous and defines (free)
dynamics on the $C^{\ast }$--algebra $\mathcal{U}$. For any realization $%
\omega \in \Omega $ and strength $\lambda \in \mathbb{R}_{0}^{+}$ of
disorder, the thermal equilibrium state of the system at inverse temperature
$\beta \in \mathbb{R}^{+}$ (i.e., $\beta >0$) is by definition the unique $%
(\tau ^{(\omega ,\lambda )},\beta )$--KMS state $\varrho ^{(\beta ,\omega
,\lambda )}$, see \cite[Example 5.3.2.]{BratteliRobinson} or \cite[Theorem
5.9]{AttalJoyePillet2006a}. It is a gauge--invariant quasi--free state which
is uniquely characterized by its symbol
\begin{equation}
\mathbf{d}_{\mathrm{fermi}}^{(\beta ,\omega ,\lambda )}:=\frac{1}{1+\mathrm{e%
}^{\beta \left( \Delta _{\mathrm{d}}+\lambda V_{\omega }\right) }}\in
\mathcal{B}(\ell ^{2}(\mathfrak{L}))  \label{Fermi statistic}
\end{equation}%
for any $\beta \in \mathbb{R}^{+}$, $\omega \in \Omega $ and $\lambda \in
\mathbb{R}_{0}^{+}$.

\section{Hilbert Space of Current Duhamel Fluctuations\label{section Current
Fluctuations copy(1)}}

We study in \cite[Theorem 4.1]{OhmIII} the rate at which resistance in the
fermion system converts electric energy into heat energy. This thermal
effect results from short range bond \emph{current fluctuations}. Note that
the relevance of the so--called algebra of normal fluctuations for transport
phenomena was observed long before \cite{OhmIII}. It is related to quantum
central limit theorems.\ See, e.g., \cite{GVV1,GVV2,GVV3,GVV4,GVV5,GVV6}.

Short range bond currents are the elements of the linear subspace
\begin{equation}
\mathcal{I}:=\mathrm{lin}\left\{ \func{Im}(a^{\ast }\left( \psi _{1}\right)
a\left( \psi _{2}\right) ):\psi _{1},\psi _{2}\in \ell ^{1}(\mathfrak{L}%
)\subset \ell ^{2}(\mathfrak{L})\right\} \subset \mathcal{U}\ .
\label{current space fluct}
\end{equation}%
As usual, $\mathrm{lin}\{\mathcal{M}\}$ denotes the linear hull of the
subset $\mathcal{M}$ of a vector space. For all $\omega \in \Omega $ and $%
\lambda \in \mathbb{R}_{0}^{+}$, the one--parameter (Bogoliubov) group $\tau
^{(\omega ,\lambda )}=\{\tau _{t}^{(\omega ,\lambda )}\}_{t\in {\mathbb{R}}}$
preserves the space $\mathcal{I}$. Indeed, the unitary group $\{\mathrm{U}%
_{t}^{(\omega ,\lambda )}\}_{t\in {\mathbb{R}}}$ (see (\ref{rescaled}) and (%
\ref{rescaledbis})) defines a strongly continuous group on $(\ell ^{1}(%
\mathfrak{L})\subset \ell ^{2}(\mathfrak{L}),\Vert \cdot \Vert _{1})$.

For any $l\in \mathbb{R}^{+}$ we define the box
\begin{equation}
\Lambda _{l}:=\{(x_{1},\ldots ,x_{d})\in \mathfrak{L}\,:\,|x_{1}|,\ldots
,|x_{d}|\leq l\}\ .  \label{eq:def lambda n}
\end{equation}%
The \emph{fluctuation observable }of the current $I\in \mathcal{I}$ is
defined by
\begin{equation}
\mathbb{F}^{(l)}\left( I\right) =\frac{1}{\left\vert \Lambda _{l}\right\vert
^{1/2}}\underset{x\in \Lambda _{l}}{\sum }\left\{ \chi _{x}\left( I\right)
-\varrho ^{(\beta ,\omega ,\lambda )}\left( \chi _{x}\left( I\right) \right)
\mathbf{1}\right\} \ ,\qquad I\in \mathcal{I}\ ,  \label{Fluctuation2bis}
\end{equation}%
where $\chi _{x}$, $x\in \mathfrak{L}$, are (space) translations, i.e., the $%
\ast $--automorphisms of $\mathcal{U}$ uniquely defined by
\begin{equation*}
\chi _{x}(a_{y})=a_{y+x}\ ,\quad y\in \mathbb{Z}^{d}\ .
\end{equation*}

We showed in \cite[Eq. (40)]{OhmIII} that the paramagnetic conductivity,
which is responsible for heat production, can be written in terms of
Green--Kubo relations involving time--correlations of \emph{bosonic} fields
coming from current fluctuations in the system. In \cite[Section 3.3]{OhmIII}
we introduced the Hilbert space of current fluctuations from $\mathbb{F}%
^{(l)}$ and the sesquilinear form on $\mathcal{U}$ naturally defined by the
state $\varrho ^{(\beta ,\omega ,\lambda )}$. This is related to the usual
construction of a GNS representation of the $(\tau ^{(\omega ,\lambda
)},\beta )$--KMS state $\varrho ^{(\beta ,\omega ,\lambda )}$.

As showed in \cite[Section A]{OhmII}, another natural GNS representation of $%
\varrho ^{(\beta ,\omega ,\lambda )}$ can be constructed from the Duhamel
two--point function defined by%
\begin{equation}
(B_{1},B_{2})_{\sim }^{(\omega )}\equiv (B_{1},B_{2})_{\sim }^{(\beta
,\omega ,\lambda )}:=\int\nolimits_{0}^{\beta }\varrho ^{(\beta ,\omega
,\lambda )}\left( B_{1}^{\ast }\tau _{i\alpha }^{(\omega ,\lambda
)}(B_{2})\right) \mathrm{d}\alpha  \label{def bogo jetmanbis}
\end{equation}%
for any $B_{1},B_{2}\in \mathcal{U}$. This \emph{positive definite}
sesquilinear form has appeared in different contexts like in linear response
theory and we recommend \cite[Section A]{OhmII} for more details. We name
this GNS representation the \emph{Duhamel GNS representation} of the $(\tau
^{(\omega ,\lambda )},\beta )$--KMS state $\varrho ^{(\beta ,\omega ,\lambda
)}$, see \cite[Definition A.6]{OhmII}. It turns out that a Hilbert space of
current fluctuations constructed from the scalar product of Duhamel GNS
representation is easier to handle and in some sense more natural.

Indeed, define the bond current observable
\begin{equation*}
I_{\mathbf{x}}:=-2\func{Im}(a_{x^{(2)}}^{\ast}a_{x^{(1)}}) \in \mathcal{I}\ ,
\end{equation*}%
for any pair $\mathbf{x}:=(x^{(1)},x^{(2)})\in \mathfrak{L}^{2}$, where $%
\left\{ \mathfrak{e}_{x}\right\} _{x\in \mathfrak{L}}$ is the canonical
orthonormal basis $\mathfrak{e}_{x}(y)\equiv \delta _{x,y}$ of $\ell ^{2}(%
\mathfrak{L})$. Then we introduce a (random) \emph{positive definite}
sesquilinear form on $\mathcal{I}$ by
\begin{equation}
(I,I^{\prime })_{\mathcal{I},l}^{(\omega )}\equiv (I,I^{\prime })_{\mathcal{I%
},l}^{(\beta ,\omega ,\lambda )}:=(\mathbb{F}^{(l)}\left( I\right) ,\mathbb{F%
}^{(l)}\left( I^{\prime }\right) )_{\sim }^{(\omega )}\ ,\qquad I,I^{\prime
}\in \mathcal{I}\ ,  \label{Fluctuation2bisbis+D}
\end{equation}%
for any $l,\beta \in \mathbb{R}^{+}$, $\omega \in \Omega $ and $\lambda \in
\mathbb{R}_{0}^{+}$.

Using \cite[Eqs.\ (24), (103)]{OhmII}, the space--averaged paramagnetic
transport coefficient
\begin{equation*}
t\mapsto \Xi _{\mathrm{p},l}^{(\omega )}\left( t\right) \equiv \Xi _{\mathrm{%
p},l}^{(\beta ,\omega ,\lambda )}\left( t\right) \in \mathcal{B}(\mathbb{R}%
^{d})
\end{equation*}%
satisfies, w.r.t. the canonical orthonormal basis $\{e_{k}\}_{k=1}^{d}$ of $%
\mathbb{R}^{d}$, the equality%
\begin{equation}
\left\{ \Xi _{\mathrm{p},l}^{(\omega )}\left( t\right) \right\} _{k,q}\equiv
\left\{ \Xi _{\mathrm{p},l}^{(\beta ,\omega ,\lambda )}\left( t\right)
\right\} _{k,q}=\Big(I_{0,e_{k}},\tau _{t}^{(\omega ,\lambda )}(I_{0,e_{q}})%
\Big)_{\mathcal{I},l}^{(\omega )}-\Big(I_{0,e_{k}},I_{0,e_{q}}\Big)_{%
\mathcal{I},l}^{(\omega )}  \label{average conductivitynew}
\end{equation}%
for any $l,\beta \in \mathbb{R}^{+}$, $\omega \in \Omega $, $\lambda \in
\mathbb{R}_{0}^{+}$, $k,q\in \{1,\ldots ,d\}$ and $t\in \mathbb{R}$. For the
basic definition of the space--averaged paramagnetic transport coefficient $%
\Xi _{\mathrm{p},l}^{(\omega )}$ we refer to \cite[Eq. (33)]{OhmII}. One may
take in this paper Equation (\ref{average conductivitynew}) as its
definition. The above expression was indeed crucial to study the
mathematical properties of $\Xi _{\mathrm{p},l}^{(\omega )}$, see \cite[%
Theorem 3.1, Corollary 3.2]{OhmII}.

Furthermore, the deterministic paramagnetic transport coefficient
\begin{equation*}
t\mapsto \mathbf{\Xi }_{\mathrm{p}}\left( t\right) \equiv \mathbf{\Xi }_{%
\mathrm{p}}^{(\beta ,\lambda )}\left( t\right) \in \mathcal{B}(\mathbb{R}%
^{d})
\end{equation*}
is defined by%
\begin{equation}
\mathbf{\Xi }_{\mathrm{p}}\left( t\right) :=\underset{l\rightarrow \infty }{%
\lim }\mathbb{E}\left[ \Xi _{\mathrm{p},l}^{(\omega )}\left( t\right) \right]
\label{paramagnetic transport coefficient macro}
\end{equation}%
for any $\beta \in \mathbb{R}^{+}$, $\lambda \in \mathbb{R}_{0}^{+}$ and $%
t\in \mathbb{R}$, see \cite[Eq. (32)]{OhmIII}. We define the limiting
positive sesquilinear form in $\mathcal{I}$ by%
\begin{equation}
(I,I^{\prime })_{\mathcal{I}}\equiv (I,I^{\prime })_{\mathcal{I}}^{(\beta
,\lambda )}:=\underset{l\rightarrow \infty }{\lim }\ \mathbb{E}\left[
(I,I^{\prime })_{\mathcal{I},l}^{(\omega )}\right] \ ,\qquad I,I^{\prime
}\in \mathcal{I}\ ,  \label{paramagnetic transport coefficient macro+1}
\end{equation}%
via the following theorem:

\begin{satz}[Sesquilinear form from current Duhamel fluctuations]
\label{toto fluctbis+D}\mbox{
}\newline
Let $\beta \in \mathbb{R}^{+}$ and $\lambda \in \mathbb{R}_{0}^{+}$. Then,
one has:\newline
\emph{(i)} The positive sesquilinear form $(\cdot ,\cdot )_{\mathcal{I}}$ is
well--defined, i.e., the limit exists:
\begin{equation*}
\underset{l\rightarrow \infty }{\lim }\ \mathbb{E}\left[ (I,I^{\prime })_{%
\mathcal{I},l}^{(\omega )}\right] \in \mathbb{R}\ ,\qquad I,I^{\prime }\in
\mathcal{I}\ .
\end{equation*}%
\emph{(ii)} There is a measurable subset $\tilde{\Omega}\equiv \tilde{\Omega}%
^{(\beta ,\lambda )}\subset \Omega $ of full measure such that, for any $%
\omega \in \tilde{\Omega}$,
\begin{equation*}
(I,I^{\prime })_{\mathcal{I}}=\underset{l\rightarrow \infty }{\lim }%
(I,I^{\prime })_{\mathcal{I},l}^{(\omega )}\ ,\qquad I,I^{\prime }\in
\mathcal{I}\ .
\end{equation*}
\end{satz}

\begin{proof}
The proof is very similar to the one of \cite[Theorem 5.26]{OhmIII}, which
concerns the (well--defined) limit
\begin{equation}
\langle I,I^{\prime }\rangle _{\mathcal{I}}\equiv \langle I,I^{\prime
}\rangle _{\mathcal{I}}^{(\beta ,\lambda )}:=\underset{l\rightarrow \infty }{%
\lim }\ \mathbb{E}\left[ \langle I,I^{\prime }\rangle _{\mathcal{I}%
,l}^{(\omega )}\right] \in \mathbb{\bar{R}}\ ,\qquad I,I^{\prime }\in
\mathcal{I}\ .  \label{eq sup necessaire1}
\end{equation}%
Here, for any $l,\beta \in \mathbb{R}^{+}$, $\omega \in \Omega $ and $%
\lambda \in \mathbb{R}_{0}^{+}$,
\begin{equation*}
\langle I,I^{\prime }\rangle _{\mathcal{I},l}^{(\omega )}\equiv \langle
I,I^{\prime }\rangle _{\mathcal{I},l}^{(\beta ,\omega ,\lambda )}:=\varrho
^{(\beta ,\omega ,\lambda )}\left( \mathbb{F}^{(l)}\left( I\right) ^{\ast }%
\mathbb{F}^{(l)}\left( I^{\prime }\right) \right) \ ,\qquad I,I^{\prime }\in
\mathcal{I}\ .
\end{equation*}%
Here, $\mathbb{F}^{(l)}$ is the fluctuation observable defined by (\ref%
{Fluctuation2bis}). In particular, one has the inequality
\begin{equation}
(I,I^{\prime })_{\mathcal{I},l}^{(\omega )}\leq \langle I,I^{\prime }\rangle
_{\mathcal{I},l}^{(\omega )}\ ,\qquad I,I^{\prime }\in \mathcal{I}\ ,
\label{auto-cor}
\end{equation}%
which results from \cite[Theorem A.4]{OhmII} for $\mathcal{X}=\mathcal{U}$
and $\varrho =\varrho ^{(\beta ,\omega ,\lambda )}$. By \cite[Lemma 5.10]%
{OhmII}, this implies the existence of a constant $D\in \mathbb{R}^{+}$ such
that, for any $l,\beta \in \mathbb{R}^{+}$, $\omega \in \Omega $, $\lambda
\in \mathbb{R}_{0}^{+}$ and all $\psi _{1},\psi _{2},\psi _{1}^{\prime
},\psi _{2}^{\prime }\in \ell ^{1}(\mathfrak{L})$,
\begin{equation}
\left\vert
\Big(%
\func{Im}(a^{\ast }\left( \psi _{1}\right) a\left( \psi _{2}\right) ),\func{%
Im}(a^{\ast }\left( \psi _{1}^{\prime }\right) a\left( \psi _{2}^{\prime
}\right) )%
\Big)%
_{\mathcal{I},l}^{(\omega )}\right\vert \leq D\left\Vert \psi
_{1}\right\Vert _{1}\left\Vert \psi _{2}\right\Vert _{1}\left\Vert \psi
_{1}^{\prime }\right\Vert _{1}\left\Vert \psi _{2}^{\prime }\right\Vert
_{1}\ .  \label{auto-cor+1}
\end{equation}%
Then, an analogue of \cite[Lemma 5.25]{OhmIII} for $(\cdot ,\cdot )_{%
\mathcal{I}}$ is proven by using the Akcoglu--Krengel ergodic theorem, see
\cite[Sections 5.2, 5.4]{OhmIII}. We omit the details since one uses very
similar arguments to those proving \cite[Theorem 5.17]{OhmIII} and the proof
is even simpler.
\end{proof}

\begin{bemerkung}[Auto--correlation upper bounds]
\label{toto fluctbis+D copy(1)}\mbox{
}\newline
The positive sesquilinear form $\langle \cdot ,\cdot \rangle _{\mathcal{I}}$
defines the Hilbert space $\mathcal{H}_{\mathrm{fl}}$ of current
fluctuations as explained in \cite[Section 3.3]{OhmIII}. By (\ref{auto-cor}%
), $(\cdot ,\cdot )_{\mathcal{I}}$ and $\langle \cdot ,\cdot \rangle _{%
\mathcal{I}}$ are related to each other via the auto--correlation upper
bounds:
\begin{equation*}
(I,I^{\prime })_{\mathcal{I}}\leq \langle I,I^{\prime }\rangle _{\mathcal{I}%
}\ ,\qquad I,I^{\prime }\in \mathcal{I}\ .
\end{equation*}
\end{bemerkung}

Hence, we define the kernel
\begin{equation*}
\mathcal{\tilde{I}}_{0}:=\left\{ I\in \mathcal{I}:(I,I)_{\mathcal{I}%
}=0\right\}
\end{equation*}%
of the positive sesquilinear form $(\cdot ,\cdot )_{\mathcal{I}}$. The
quotient $\mathcal{I}/\mathcal{\tilde{I}}_{0}$ is a pre--Hilbert space and
its completion w.r.t. the scalar product
\begin{equation}
([I],[I^{\prime }])_{\mathcal{I}/\mathcal{\tilde{I}}_{0}}:=\ (I,I^{\prime
})_{\mathcal{I}}\ ,\qquad \lbrack I],[I^{\prime }]\in \mathcal{I}/\mathcal{%
\tilde{I}}_{0}\ ,  \label{Fluctuation3+D}
\end{equation}%
is the Hilbert space%
\begin{equation}
\left( \mathcal{\tilde{H}}_{\mathrm{fl}},(\cdot ,\cdot )_{\mathcal{\tilde{H}}%
_{\mathrm{fl}}}\right)  \label{Fluctuation4+D}
\end{equation}%
of current \emph{Duhamel} fluctuations. The dynamics defined by $\tau
^{(\omega ,\lambda )}=\{\tau _{t}^{(\omega ,\lambda )}\}_{t\in {\mathbb{R}}}$
on $\mathcal{U}$ induces a unitary time evolution on $\mathcal{\tilde{H}}_{%
\mathrm{fl}}$:

\begin{satz}[Dynamics of current Duhamel fluctuations]
\label{bound incr 1 Lemma copy(1)+D}\mbox{
}\newline
Let $\beta \in \mathbb{R}^{+}$ and $\lambda \in \mathbb{R}_{0}^{+}$. Then,
there is a measurable subset $\tilde{\Omega}\equiv \tilde{\Omega}^{(\beta
,\lambda )}\subset \Omega $ of full measure such that, for any $\omega \in
\tilde{\Omega}$, there is a unique, strongly continuous one--parameter
unitary group $\{\mathrm{\tilde{V}}_{t}^{(\omega ,\lambda )}\}_{t\in {%
\mathbb{R}}}$ on the Hilbert space $\mathcal{\tilde{H}}_{\mathrm{fl}}$
obeying, for any $t\in {\mathbb{R}}$,
\begin{equation*}
\mathrm{\tilde{V}}_{t}^{(\omega ,\lambda )}([I])=[\tau _{t}^{(\omega
,\lambda )}(I)]\ ,\qquad \lbrack I]\in \mathcal{I}/\mathcal{\tilde{I}}_{0}\ .
\end{equation*}
\end{satz}

\begin{proof}
The proof is essentially the same as the one of \cite[Theorem 5.27]{OhmIII}.
We omit the details. Note that one uses (\ref{auto-cor})--(\ref{auto-cor+1})
combined with \cite[Corollary A.8]{OhmII}.
\end{proof}

\begin{bemerkung}[Deterministic unitary group]
\mbox{
}\newline
As in \cite[Section 5.5.3]{OhmIII}, by using the Duhamel representation \cite%
[Definition A.6]{OhmII} one can construct a unique, strongly continuous
one--parameter deterministic unitary group $\{\mathrm{\hat{V}}_{t}^{(\lambda
)}\}_{t\in {\mathbb{R}}}$ on a direct integral Hilbert space.
\end{bemerkung}

By using the Hilbert space $\mathcal{\tilde{H}}_{\mathrm{fl}}$ (\ref%
{Fluctuation4+D}) of current Duhamel fluctuations, we infer from Equations (%
\ref{average conductivitynew}) and (\ref{paramagnetic transport coefficient
macro})--(\ref{paramagnetic transport coefficient macro+1}) that
\begin{equation}
\left\{ \mathbf{\Xi }_{\mathrm{p}}\left( t\right) \right\} _{k,q}=%
\Big(%
[I_{0,e_{k}}],\mathrm{\tilde{V}}_{t}^{(\omega ,\lambda )}([I_{0,e_{q}}])%
\Big)%
_{\mathcal{\tilde{H}}_{\mathrm{fl}}}-%
\Big(%
[I_{0,e_{k}}],[I_{0,e_{q}}]%
\Big)%
_{\mathcal{\tilde{H}}_{\mathrm{fl}}}  \label{cesaro0}
\end{equation}%
for any $\beta \in \mathbb{R}^{+}$, $\lambda \in \mathbb{R}_{0}^{+}$, $%
k,q\in \{1,\ldots ,d\}$ and all $t\in \mathbb{R}$. Here, $\omega $ belongs
to some measurable subset of full measure defined such that the strongly
continuous one--parameter unitary group $\{\mathrm{\tilde{V}}_{t}^{(\omega
,\lambda )}\}_{t\in {\mathbb{R}}}$ exists, see Theorem \ref{bound incr 1
Lemma copy(1)+D}. Equation (\ref{cesaro0}) is the analogue of (\ref{average
conductivitynew}) for $\Xi _{\mathrm{p},l}^{(\omega )}$. As a consequence,
we can now follow the same strategy as in \cite[Section 5.1.2]{OhmII}. This
is performed in the next section.

\section{Macroscopic Conductivity of Fermion Systems\label{Sect Conductivity
of Fermion Systems}}

As in \cite[Definition 3.2]{OhmIII}, for any $\beta \in \mathbb{R}^{+}$ and $%
\lambda \in \mathbb{R}_{0}^{+}$, the macroscopic conductivity is the map
\begin{equation}
t\mapsto \mathbf{\Sigma }\left( t\right) \equiv \mathbf{\Sigma }^{(\beta
,\lambda )}\left( t\right) :=\left\{
\begin{array}{lll}
0 & , & \qquad t\leq 0\ . \\
\mathbf{\Xi }_{\mathrm{d}}+\mathbf{\Xi }_{\mathrm{p}}\left( t\right) & , &
\qquad t\geq 0\ .%
\end{array}%
\right.  \label{conductivity}
\end{equation}%
Here, $\mathbf{\Xi }_{\mathrm{p}}$ is the deterministic paramagnetic
transport coefficient defined by (\ref{paramagnetic transport coefficient
macro}), whereas the time--independent operator $\mathbf{\Xi }_{\mathrm{d}%
}\in \mathcal{B}(\mathbb{R}^{d})$ is the diamagnetic transport coefficient,
which equals%
\begin{equation}
\left\{ \mathbf{\Xi }_{\mathrm{d}}\right\} _{k,q}=2\delta _{k,q}\func{Re}%
\left\{ \mathbb{E}\left[ \langle \mathfrak{e}_{e_{k}},\mathbf{d}_{\mathrm{%
fermi}}^{(\beta ,\omega ,\lambda )}\mathfrak{e}_{0}\rangle \right] \right\}
\label{explicit form conductivity}
\end{equation}%
for any $\beta \in \mathbb{R}^{+}$, $\lambda \in \mathbb{R}_{0}^{+}$ and $%
k,q\in \{1,\ldots ,d\}$, see \cite[Eq. (37)]{OhmII}. $\langle \cdot ,\cdot
\rangle $ is the scalar product in $\ell ^{2}(\mathfrak{L})$ and recall that
the positive bounded operator $\mathbf{d}_{\mathrm{fermi}}^{(\beta ,\omega
,\lambda )}$ is defined by (\ref{Fermi statistic}).

Since we assume the random potential to be i.i.d. the paramagnetic and
diamagnetic transport coefficients turn out to be both a multiple of the
identity, see \cite[Eqs. (68)--69)]{OhmIII}. In particular, there is a
function
\begin{equation*}
\mathbf{\sigma }_{\mathrm{p}}\equiv \mathbf{\sigma }_{\mathrm{p}}^{(\beta
,\lambda )}\in C(\mathbb{R};\mathbb{R}_{0}^{-})
\end{equation*}%
and a constant $\mathbf{\sigma }_{\mathrm{d}}\equiv \mathbf{\sigma }_{%
\mathrm{d}}^{(\beta ,\lambda )}$ such that
\begin{equation}
\mathbf{\Xi }_{\mathrm{p}}\left( t\right) =\mathbf{\sigma }_{\mathrm{p}%
}\left( t\right) \ \mathrm{Id}_{\mathbb{R}^{d}}\ ,\qquad \mathbf{\Xi }_{%
\mathrm{d}}=\mathbf{\sigma }_{\mathrm{d}}\ \mathrm{Id}_{\mathbb{R}^{d}}\ ,
\label{diagoanal coeef}
\end{equation}%
for any $\beta \in \mathbb{R}^{+}$, $\lambda \in \mathbb{R}_{0}^{+}$ and $%
t\in \mathbb{R}$. Note additionally that, for all $t\in \mathbb{R}$, $%
\mathbf{\sigma }_{\mathrm{p}}(t)=\mathbf{\sigma }_{\mathrm{p}}(|t|)$ with $%
\mathbf{\sigma }_{\mathrm{p}}(0)=0$ and
\begin{equation*}
\mathbf{\sigma }_{\mathrm{p}}(t)\in \lbrack -2\Vert \lbrack
I_{0,e_{1}}]\Vert _{\mathcal{\tilde{H}}_{\mathrm{fl}}}^{2},0]\ ,
\end{equation*}%
see (\ref{cesaro0}). Thus the \emph{in--phase} conductivity of the fermion
system equals
\begin{equation}
\mathbf{\sigma }(t)\equiv \mathbf{\sigma }^{(\beta ,\lambda )}(t):=\mathbf{%
\sigma }_{\mathrm{p}}(t)+\mathbf{\sigma }_{\mathrm{d}}\ ,\qquad t\in \mathbb{%
R}\ .  \label{in phase conductivity}
\end{equation}%
Clearly, $\mathbf{\sigma }\in C(\mathbb{R};\mathbb{R})$ satisfies $\mathbf{%
\sigma }(t)=\mathbf{\sigma }(-t)$ with $\mathbf{\sigma }(0)=\mathbf{\sigma }%
_{\mathrm{d}}$. Since the diamagnetic conductivity $\mathbf{\sigma }_{%
\mathrm{d}}$ is an explicit constant, that is,
\begin{equation}
\mathbf{\sigma }_{\mathrm{d}}=2\func{Re}\left\{ \mathbb{E}\left[ \langle
\mathfrak{e}_{e_{1}},\mathbf{d}_{\mathrm{fermi}}^{(\beta ,\omega ,\lambda )}%
\mathfrak{e}_{0}\rangle \right] \right\} \ ,  \label{sigmad calcul}
\end{equation}%
the study of the in--phase conductivity $\mathbf{\sigma }$ corresponds to
the analysis of the properties of $\mathbf{\sigma }_{\mathrm{p}} $. We
follow the same strategy as in \cite[Section 5.1.2]{OhmII}.

First, we denote by $i\mathcal{\tilde{L}}_{\mathrm{fl}}^{(\omega )}$ the
anti--self--adjoint operator acting on $\mathcal{\tilde{H}}_{\mathrm{fl}}$
generating the unitary group $\{\mathrm{\tilde{V}}_{t}^{(\omega ,\lambda
)}\}_{t\in {\mathbb{R}}}$ of Theorem \ref{bound incr 1 Lemma copy(1)+D}.
Then, one deduces from Equation (\ref{cesaro0}) and the spectral theorem the
existence of the paramagnetic conductivity measure $\mu _{\mathrm{p}}$, like
in \cite[Theorem 5.4]{OhmII}:

\begin{satz}[Paramagnetic conductivity measures]
\label{lemma sigma pos type copy(4)-macro}\mbox{
}\newline
Let $\beta \in \mathbb{R}^{+}$and $\lambda \in \mathbb{R}_{0}^{+}$. Then,
there is a positive symmetric measure $\mu _{\mathrm{p}}\equiv \mu _{\mathrm{%
p}}^{(\beta ,\lambda )}$ on $\mathbb{R}$ such that $\mu _{\mathrm{p}}\left(
\mathbb{R}\right) <\infty $ uniformly w.r.t. $\beta \in \mathbb{R}^{+}$, $%
\lambda \in \mathbb{R}_{0}^{+}$, while%
\begin{equation}
\mathbf{\sigma }_{\mathrm{p}}(t)=\int_{\mathbb{R}}\left( \cos \left( t\nu
\right) -1\right) \mu _{\mathrm{p}}(\mathrm{d}\nu )\ ,\qquad t\in \mathbb{R}%
\ .  \label{sigma - measure}
\end{equation}
\end{satz}

\begin{proof}
As explained above, the existence of the finite positive symmetric measure $%
\mu _{\mathrm{p}}$ on $\mathbb{R}$ satisfying (\ref{sigma - measure}) is a
consequence of the spectral theorem applied to $i\mathcal{\tilde{L}}_{%
\mathrm{fl}}^{(\omega )}$ together with $\mathbf{\sigma }_{\mathrm{p}}(t)=%
\mathbf{\sigma }_{\mathrm{p}}(|t|)$ and $\mathbf{\sigma }_{\mathrm{p}}(0)=0$%
. See Equations (\ref{cesaro0}) and (\ref{diagoanal coeef}). Observe also
that $\mu _{\mathrm{p}}$ is clearly a deterministic measure. Moreover,
\begin{equation*}
\mu _{\mathrm{p}}\left( \mathbb{R}\right) =\left(
[I_{0,e_{1}}],[I_{0,e_{1}}]\right) _{\mathcal{\tilde{H}}_{\mathrm{fl}}}
\end{equation*}%
and we thus deduce from (\ref{auto-cor+1}) that this quantity is uniformly
bounded w.r.t. $\beta \in \mathbb{R}^{+}$ and $\lambda \in \mathbb{R}%
_{0}^{+} $.
\end{proof}

\begin{bemerkung}[On the strict negativity of the paramagnetic conductivity]

\mbox{
}\newline
In contrast to the standard Liouvillian $\mathcal{\tilde{L}}$ in \cite[Eq.
(105)]{OhmII}, it is a priori not clear whether the kernel of $\mathcal{%
\tilde{L}}_{\mathrm{fl}}^{(\omega )}$ is empty or not. Thus, we define $%
\mathrm{P}_{\mathrm{fl}}^{(\omega )}$ to be the orthogonal projection on the
kernel of $\mathcal{\tilde{L}}_{\mathrm{fl}}^{(\omega )}$. By (\ref{cesaro0}%
) and (\ref{diagoanal coeef}) combined with the stationarity of KMS states,
one can prove that $\mathbf{\sigma }_{\mathrm{p}}\left( t\right) =0$ for $%
t\neq 0$ iff $\mathrm{P}_{\mathrm{fl}}^{(\omega
)}[I_{0,e_{1}}]=[I_{0,e_{1}}] $. In particular, if $\mathbf{\sigma }_{%
\mathrm{p}}\left( t\right) =0$ for some $t\in \mathbb{R}\backslash \{0\}$
then $\mathbf{\sigma }_{\mathrm{p}}$ is the zero function on $\mathbb{R}$
.In the same way, if there is $t\in {\mathbb{R}}$ where $\mathbf{\sigma }_{%
\mathrm{p}}\left( t\right) \neq 0$\ then $\mathbf{\sigma }_{\mathrm{p}%
}\left( t\right) <0$ for all $t\in \mathbb{R}\backslash \{0\}$.
\end{bemerkung}

Note that Theorem \ref{lemma sigma pos type copy(4)-macro} is a reminiscent
of \cite[Theorem 3.1 (v)]{OhmII} where we show the existence of a local
paramagnetic conductivity measure $\mu _{\mathrm{p},l}^{(\omega )}\equiv \mu
_{\mathrm{p},l}^{(\beta ,\omega ,\lambda )}$. It is a positive operator
valued measure that satisfies
\begin{equation*}
\int_{\mathbb{R}}\left( 1+\left\vert \nu \right\vert \right) \Vert \mu _{%
\mathrm{p},l}^{(\omega )}\Vert _{\mathrm{op}}(\mathrm{d}\nu )<\infty \ ,
\end{equation*}%
uniformly w.r.t. $l,\beta \in \mathbb{R}^{+}$, $\omega \in \Omega $, $%
\lambda \in \mathbb{R}_{0}^{+}$, and%
\begin{equation*}
\Xi _{\mathrm{p},l}^{(\omega )}(t)=\int_{\mathbb{R}}\left( \cos \left( t\nu
\right) -1\right) \mu _{\mathrm{p},l}^{(\omega )}(\mathrm{d}\nu )\ ,\qquad
t\in \mathbb{R}\ .
\end{equation*}%
Recall that $\Xi _{\mathrm{p},l}^{(\omega )}$ is the space--averaged
paramagnetic transport coefficient, see (\ref{average conductivitynew}). For
all $l\in \mathbb{R}^{+}$, the map $\omega \mapsto \mu _{\mathrm{p}%
,l}^{(\omega )}$ is measurable w.r.t. the $\sigma $--algebra $\mathfrak{A}%
_{\Omega }$ and the weak$^*$ topology for ($\mathcal{B}(\mathbb{R}^d)$%
--valued) measures on $\mathbb{R}$. $\mathbb{E}[\mu _{\mathrm{p},l}^{(\omega
)}]$, seen as a weak integral, is a finite positive measure. Indeed, as $l
\to \infty$, it converges to the positive measure $\mu _{\mathrm{p}}$ $%
\mathrm{Id}_{\mathbb{R}^{d}}$, with $\mu_\mathrm{p}$ as in Theorem \ref%
{lemma sigma pos type copy(4)-macro}:

\begin{satz}[From microscopic to macroscopic conductivity measures]
\label{lemma sigma pos type copy(5)}\mbox{
}\newline
Let $\beta \in \mathbb{R}^{+}$ and $\lambda \in \mathbb{R}_{0}^{+}$. Then
there is a measurable set $\tilde{\Omega} \equiv \tilde{\Omega}^{%
(\beta,\lambda)}\subset \Omega$ of full measure such that, for all $\omega
\in \tilde{\Omega}$, $\mu _{\mathrm{p},l}^{(\omega )}$ converges in the weak$%
^{\ast }$--topology to $\mu _{\mathrm{p}}\ \mathrm{Id}_{\mathbb{R}^{d}}$, as
$l \to \infty$. In particular, $\mathbb{E}[\mu _{\mathrm{p},l}^{(\omega )}]$
converges in the weak$^{\ast }$--topology to $\mu _{\mathrm{p}}\ \mathrm{Id}%
_{\mathbb{R}^{d}}$, as $l \to \infty$.
\end{satz}

\begin{proof}
The limit in \cite[Theorem 3.1 (p)]{OhmIII} is uniform w.r.t. times $t$ in
compact sets. This implies the weak$^{\ast }$--convergence of $\mu _{\mathrm{%
p},l}^{(\omega )}$ towards $\mu _{\mathrm{p}}\ \mathrm{Id}_{\mathbb{R}^{d}}$
for $\omega$ in a measurable set $\tilde{\Omega}^{(\beta,\lambda)}
\subset \Omega$ of full measure.
\end{proof}

\begin{koro}[First moment of the paramagnetic conductivity measure]
\label{lemma sigma pos type}\mbox{
}\newline
Let $\beta \in \mathbb{R}^{+}$and $\lambda \in \mathbb{R}_{0}^{+}$. Then,
\begin{equation*}
\int_{\mathbb{R}}\left( 1+\left\vert \nu \right\vert \right) \mu _{\mathrm{p}%
}(\mathrm{d}\nu )<\infty \ ,
\end{equation*}%
uniformly w.r.t. $\beta \in \mathbb{R}^{+}$ and $\lambda \in \mathbb{R}%
_{0}^{+}$. In particular, the family $\{\mathbf{\sigma }_{\mathrm{p}%
}^{(\beta ,\lambda )}\}_{\beta \in \mathbb{R}^{+},\lambda \in \mathbb{R}%
_{0}^{+}}$ of maps from $\mathbb{R}$ to $\mathbb{R}_{0}^{-}$ is
equicontinuous.
\end{koro}

\begin{proof}
By Theorem \ref{lemma sigma pos type copy(4)-macro}, it suffices to prove
that
\begin{equation*}
\int_{0}^{\infty }\nu \ \mu _{\mathrm{p}}(\mathrm{d}\nu )<\infty \ ,
\end{equation*}%
uniformly w.r.t. $\beta \in \mathbb{R}^{+}$ and $\lambda \in \mathbb{R}%
_{0}^{+}$. By using Theorem \ref{lemma sigma pos type copy(5)} and \cite[%
Theorem 5.5]{OhmII}, we arrive at
\begin{equation*}
\underset{\nu _{0}\rightarrow \infty }{\lim }\int_{0}^{\nu _{0}}\nu \ \mu _{%
\mathrm{p}}(\mathrm{d}\nu )\leq 2\langle \lbrack
I_{0,e_{1}}],[I_{0,e_{1}}]\rangle _{\mathcal{I}}\ .
\end{equation*}%
Combined with \cite[Lemma 5.10]{OhmII}, this implies the existence of a
constant $D\in \mathbb{R}^{+}$ not depending on $\beta \in \mathbb{R}^{+}$
and $\lambda \in \mathbb{R}_{0}^{+}$ such that%
\begin{equation*}
\underset{\nu _{0}\rightarrow \infty }{\lim }\int_{0}^{\nu _{0}}\nu \ \mu _{%
\mathrm{p}}(\mathrm{d}\nu )\leq D<\infty \ .
\end{equation*}%
Since $\mu _{\mathrm{p}}$ is a positive measure, the above limit exists and
the equicontinuity of the paramagnetic conductivity is deduced like in the
proof of \cite[Corollary 3.2 (iv)]{OhmII}.
\end{proof}

Note that the diamagnetic conductivity $\mathbf{\sigma }_{\mathrm{d}}$ is
constant in time and its Fourier transform is the atomic measure $\mathbf{%
\sigma }_{\mathrm{d}}\delta _{0}$, see (\ref{explicit form conductivity}).
Since the conductivity $\Sigma $ (\ref{conductivity}) is the sum of the
paramagnetic and diamagnetic conductivities, we define the \emph{in--phase\ }%
conducti%
\-%
vity measure $\mu _{\mathbf{\Sigma }}\equiv \mu _{\mathbf{\Sigma }^{(\beta
,\lambda )}}$ by
\begin{equation}
\mu _{\mathbf{\Sigma }}:=\mu _{\mathrm{p}}+\left( \mathbf{\sigma }_{\mathrm{d%
}}-\mu _{\mathrm{p}}\left( \mathbb{R}\right) \right) \delta _{0}
\label{equ super cond}
\end{equation}%
for any $\beta \in \mathbb{R}^{+}$ and $\lambda \in \mathbb{R}_{0}^{+}$. By
Theorem \ref{lemma sigma pos type copy(4)-macro}, the in--phase conductivity
$\mathbf{\sigma }$\ given in (\ref{in phase conductivity}), equals
\begin{equation*}
\mathbf{\sigma }(t)=\int_{\mathbb{R}}\cos (t\nu )\mu _{\mathbf{\Sigma }}(%
\mathrm{d}\nu )=\int_{\mathbb{R}}\left( \cos (t\nu )-1\right) \mu _{\mathrm{p%
}}(\mathrm{d}\nu )+\mathbf{\sigma }_{\mathrm{d}}\ ,\qquad t\in \mathbb{R}\ .
\end{equation*}%
The restricted measure $\mu _{\mathrm{AC}}:=\mu _{\mathrm{p}}|_{\mathbb{R}%
\backslash \{0\}}$ is the (in--phase) \emph{AC}--conducti%
\-%
vity measure described in \cite[Theorem 4.4]{OhmIII}, which was deduced from
the second principle of thermodynamics. The additional information we obtain
here is the finiteness of $\mu _{\mathrm{p}}$, i.e., $\mu _{\mathrm{p}}(%
\mathbb{R})<\infty $. Note that the atomic measure at $\nu =0$
\begin{equation*}
\mu _{0}:=\left( \mathbf{\sigma }_{\mathrm{d}}-\mu _{\mathrm{p}}\left(
\mathbb{R}\backslash \{0\}\right) \right) \delta _{0}\
\end{equation*}%
does not a priori vanish.

\begin{bemerkung}[On the strict negativity of the paramagnetic conductivity]

\mbox{
}\newline
Similar to \cite[Theorem 5.9]{OhmII}, the conductivity measure $\mu _{%
\mathrm{AC}}$ can be reconstructed from some macroscopic quantum current
viscosity. We refrain from doing it here.
\end{bemerkung}

Note that the case $\lambda =0$ can be interpreted as the perfect conductor.
Indeed, by explicit computations using the dispersion relation
\begin{equation}
E(p):=2\left[ d-\left( \cos (p_{1})+\cdots +\cos (p_{d})\right) \right] \
,\qquad p\in \left[ -\pi ,\pi \right] ^{d}\ ,
\label{dispertion relation section2}
\end{equation}%
of the (up to a minus sign) discrete Laplacian $\Delta _{\mathrm{d}}$,
\begin{equation*}
\left\langle \mathfrak{e}_{x},\mathbf{d}_{\mathrm{fermi}}^{(\beta ,\omega ,
0 )}\mathfrak{e}_{y}\right\rangle =\frac{1}{(2\pi )^{d}}\int\nolimits_{[-\pi
,\pi ]^{d}}\frac{1}{1+\mathrm{e}^{\beta E(p)}}\mathrm{e}^{-ip\cdot (x-y)}%
\mathrm{d}^{d}p\ ,
\end{equation*}
we obtain
\begin{equation}
\mathbf{\sigma }_{\mathrm{d}}^{(\beta ,0)}=\frac{2}{(2\pi )^{d}}%
\int\nolimits_{[-\pi ,\pi ]^{d}}\frac{\cos \left( p_{1}\right) }{1+\mathrm{e}%
^{\beta E(p)}}\mathrm{d}^{d}p \not=0  \label{explicit constant}
\end{equation}%
for any $\beta \in \mathbb{R}^{+}$, whereas $\mu _{\mathrm{p}}^{(\beta ,0)}(%
\mathbb{R})=0$ (cf. Lemma \ref{lemma conductivty4 copy(1)}). Hence, the heat
production vanishes in this special case. Similarly, the limit $\lambda
\rightarrow \infty $ corresponds to the perfect insulator and also gives a
vanishing heat production for any cyclic processes involving the external
electromagnetic field:

\begin{satz}[Conductivity -- Asymptotics]
\label{main 4}\mbox{
}\newline
Let $\beta \in \mathbb{R}^{+}$ and assume that $\mathfrak{a}_{\mathbf{0}}$
is absolutely continuous w.r.t. the Lebesgue measure when we perform the
limit $\lambda \rightarrow \infty $.\newline
\emph{(p)} Paramagnetic conductivity: $\mathbf{\sigma }_{\mathrm{p}}^{(\beta
,\lambda )}(t)$ converges uniformly on compact sets to zero, as $\lambda
\rightarrow 0^{+}$ or $\lambda \rightarrow \infty $. In particular, $\mu _{%
\mathrm{p}}$ converges in the weak$^{\ast }$--topology to the trivial
measure in these two cases. \newline
\emph{(d)} Diamagnetic conductivity: $\mathbf{\sigma }_{\mathrm{d}}^{(\beta
,\lambda )}$ converges to $\mathbf{\sigma }_{\mathrm{d}}^{(\beta ,0)}$, as $%
\lambda \rightarrow 0^{+}$, and to zero, as $\lambda \rightarrow \infty $.
\end{satz}

\begin{proof}
(p) The assertions follow from Proposition \ref{main 1 copy(5)} and Lemma %
\ref{lemma conductivty4 copy(1)}.

\noindent (d) The corresponding assertions for $\mathbf{\sigma }_{\mathrm{d}%
} $ can be shown by using the same kind of (explicit) computation as for $%
\mathbf{\sigma }_{\mathrm{p}}$ and are even much simpler to prove than for
the paramagnetic case. Indeed, they follow from (\ref{sigmad calcul}) and
direct estimates: To study the limit $\lambda \rightarrow 0^{+}$, use (\ref%
{Duhamel's formula encore}) to get that, for any $\beta ,\lambda \in \mathbb{%
R}^{+}$,
\begin{equation*}
\left\Vert \mathbf{d}_{\mathrm{fermi}}^{(\beta ,\omega ,\lambda )}-\mathbf{d}%
_{\mathrm{fermi}}^{(\beta ,\omega ,0)}\right\Vert _{\mathrm{op}}\leq
\left\Vert \mathrm{e}^{\beta \Delta _{\mathrm{d}}}-\mathrm{e}^{\beta \left(
\Delta _{\mathrm{d}}+\lambda V_{\omega }\right) }\right\Vert _{\mathrm{op}%
}\leq \beta \mathrm{e}^{2d\beta }\left\vert \lambda \right\vert \ .
\end{equation*}%
Under the condition that $\mathfrak{a}_{\mathbf{0}}$ is absolutely
continuous w.r.t. the Lebesgue measure, by a similar but easier computation
using Duhamel expansions as done in Section \ref{Section Asymptotics}, one
verifies that
\begin{equation*}
\underset{\lambda \rightarrow \infty }{\lim }\mathbb{E}\left[ \langle
\mathfrak{e}_{e_1}, \mathbf{d}_{\mathrm{fermi}}^{(\beta ,\omega ,\lambda )}
\mathfrak{e}_0 \rangle \right ] =0\ .
\end{equation*}%
This shows the case $\lambda \to \infty$, by Equation (\ref{sigmad calcul})
\end{proof}

By the second principle of thermodynamics, the fermion system cannot
transfer any energy to the electromagnetic field. In fact, the fermion
system even absorbs, in general, some non--vanishing amount of
electromagnetic energy in form of heat. To explain this, let $\mathcal{S}(%
\mathbb{R\times R}^{d};\mathbb{R}^{d})$ be the Fr\'{e}chet space of Schwartz
functions $\mathbb{R\times R}^{d}\rightarrow \mathbb{R}^{d}$ endowed with
the usual locally convex topology. The electromagnetic potential is here an
element $\mathbf{A}\in C_{0}^{\infty }(\mathbb{R}\times {\mathbb{R}}^{d};{%
\mathbb{R}}^{d})\subset \mathcal{S}(\mathbb{R\times R}^{d};\mathbb{R}^{d}) $
and the electric field equals%
\begin{equation}
E_{\mathbf{A}}(t,x):=-\partial _{t}\mathbf{A}(t,x)\ ,\quad t\in \mathbb{R},\
x\in \mathbb{R}^{d}\ .  \label{V bar 0}
\end{equation}%
Then one gets:

\begin{satz}[Absorption of electromagnetic energy]
\label{main 2 copy(1)}\mbox{
}\newline
Let $\lambda _{0}\in \mathbb{R}^{+}$. Then there is $\beta _{0}\in \mathbb{R}%
^{+}$ such that, for any $\beta \in (0,\beta _{0})$ and $\lambda \in
(\lambda _{0}/2,\lambda _{0})$,%
\begin{equation*}
\mu _{\mathrm{AC}}\left( \mathbb{R}\backslash \{0\}\right) >0\ .
\end{equation*}%
Equivalently, there is a meager set $\mathcal{Z}\subset C_{0}^{\infty }(%
\mathbb{R}\times {\mathbb{R}}^{d};{\mathbb{R}}^{d}) \subset \mathcal{S}(%
\mathbb{R\times R}^{d};\mathbb{R}^{d}) $ such that, for all $\mathbf{A}\in
C_{0}^{\infty }(\mathbb{R}\times {\mathbb{R}}^{d};{\mathbb{R}}%
^{d})\backslash \mathcal{Z}$,%
\begin{equation*}
\int\nolimits_{\mathbb{R}}\mathrm{d}s_{1}\int\nolimits_{\mathbb{R}}\mathrm{d}%
s_{2}\ \mathbf{\Sigma }(s_{1}-s_{2})\int\nolimits_{\mathbb{R}^{d}}\mathrm{d}%
^{d}x\left\langle E_{\mathbf{A}}(s_{2},x),E_{\mathbf{A}}(s_{1},x)\right%
\rangle >0\ .
\end{equation*}
\end{satz}

\begin{proof}
Use Lemmata \ref{main 2 copy(2)} and \ref{lemma conductivty4 copy(7)2}.
\end{proof}

It means that the paramagnetic conductivity $\mathbf{\sigma }_{\mathrm{p}}$
is generally non--zero and thus causes a strictly positive heat production
for non--vanishing electric fields. This is the case of usual conductors.

\section{Technical Proofs\label{Sect tehnical conduc}}

We gather here some technical assertions used to prove Theorems \ref{main 4}%
--\ref{main 2 copy(1)}. We divide the section in two parts. The first
subsection is a study of asymptotic properties of the paramagnetic
conductivity, whereas the second one is a proof that the fermion system
generally absorbs a non--vanishing amount of electromagnetic work in form of
heat.

Before starting our proofs, we recall some definitions used in \cite%
{OhmII,OhmIII}: First, $C_{t+i\alpha }^{(\omega )}$ is the complex--time
two--point correlation function, see \cite[Section 5.1]{OhmIII} for more
details. For all $\beta \in \mathbb{R}^{+}$, $\omega \in \Omega $, $\lambda
\in \mathbb{R}_{0}^{+}$, $t\in {\mathbb{R}}$ and $\alpha \in \lbrack 0,\beta
]$, it equals
\begin{equation}
C_{t+i\alpha }^{(\omega )}(\mathbf{x})=\langle \mathfrak{e}_{x^{(2)}},%
\mathrm{e}^{-it\left( \Delta _{\mathrm{d}}+\lambda V_{\omega }\right)
}F_{\alpha }^{\beta }\left( \Delta _{\mathrm{d}}+\lambda V_{\omega }\right)
\mathfrak{e}_{x^{(1)}}\rangle \ ,\quad \mathbf{x}:=(x^{(1)},x^{(2)})\in
\mathfrak{L}^{2}\ ,  \label{cond two--point correlation function}
\end{equation}%
where the real function $F_{\alpha }^{\beta }$ is defined, for any $\beta
\in \mathbb{R}^{+}$ and $\alpha \in {\mathbb{R}}$, by
\begin{equation}
F_{\alpha }^{\beta }\left( \varkappa \right) :=\frac{\mathrm{e}^{\alpha
\varkappa }}{1+\mathrm{e}^{\beta \varkappa }}\ ,\qquad \varkappa \in {%
\mathbb{R}}\ .  \label{function F}
\end{equation}%
Then we set for any $\beta \in \mathbb{R}^{+}$, $\omega \in \Omega $, $%
\lambda \in \mathbb{R}_{0}^{+}$, $t\in {\mathbb{R}}$, $\alpha \in \lbrack
0,\beta ]$, $\mathbf{x}:=(x^{(1)},x^{(2)})\in \mathfrak{L}^{2}$ and $\mathbf{%
y}:=(y^{(1)},y^{(2)})\in \mathfrak{L}^{2}$,%
\begin{equation}
\mathfrak{C}_{t+i\alpha }^{(\omega )}(\mathbf{x},\mathbf{y})=\underset{\pi
,\pi ^{\prime }\in S_{2}}{\sum }\varepsilon _{\pi }\varepsilon _{\pi
^{\prime }}C_{t+i\alpha }^{(\omega )}(y^{\pi ^{\prime }(1)},x^{\pi
(1)})C_{-t+i(\beta -\alpha )}^{(\omega )}(x^{\pi (2)},y^{\pi ^{\prime
}(2)})\ ,  \label{map coolbis}
\end{equation}%
compare with \cite[Eq. (93)]{OhmII}. Here, $\pi ,\pi ^{\prime }\in S_{2}$
are by definition permutations of $\{1,2\}$ with signatures $\varepsilon
_{\pi },\varepsilon _{\pi ^{\prime }}\in \{-1,1\}$. In \cite[Eq. (141)]%
{OhmIII} we define the function
\begin{equation}
\Gamma _{1,1}(t):=\underset{l\rightarrow \infty }{\lim }\frac{1}{\left\vert
\Lambda _{l}\right\vert }\sum\limits_{x,y\in \Lambda _{l}}\mathbb{E}\left[
\int\nolimits_{0}^{\beta }\mathfrak{C}_{t+i\alpha }^{(\omega
)}(x,x-e_{1},y,y-e_{1})\mathrm{d}\alpha \right]  \label{positivity eq2}
\end{equation}%
and, by \cite[Eq. (147)]{OhmIII}, observe that%
\begin{equation}
\mathbf{\sigma }_{\mathrm{p}}(t)=\Gamma _{1,1}(t)-\Gamma _{1,1}(0)
\label{positivity eq1}
\end{equation}%
for any $\beta \in \mathbb{R}^{+}$, $\lambda \in \mathbb{R}_{0}^{+}$ and $%
t\in \mathbb{R}$. Now we are ready to prove Theorems \ref{main 4} and \ref%
{main 2 copy(1)}.

\subsection{Asymptotics of Paramagnetic Conductivity\label{Section
Asymptotics}}

Here we study the asymptotic properties of the paramagnetic conductivity $%
\mathbf{\sigma }_{\mathrm{p}}$, as $\lambda \rightarrow 0^{+}$ and $\lambda
\rightarrow \infty $. In other words, we prove Theorem \ref{main 4} (p). We
break this proof in several lemmata and one proposition.

By (\ref{positivity eq1}) and \cite[Lemma 5.16]{OhmIII}, for any $%
\varepsilon ,\beta \in \mathbb{R}^{+}$, $\lambda \in \mathbb{R}_{0}^{+}$ and
$\upsilon \in (0,\beta /2)$,%
\begin{equation}
\mathbf{\sigma }_{\mathrm{p}}\left( t\right) =4d\left( \tilde{\Gamma}%
_{\upsilon ,\varepsilon ,1,1}(t)-\tilde{\Gamma}_{\upsilon ,\varepsilon
,1,1}(0)\right) +\mathcal{O}(\upsilon )+\mathcal{O}_{\upsilon }(\varepsilon
)\ ,  \label{important estimate conduc}
\end{equation}%
uniformly for times $t$ in compact sets. The term of order $\mathcal{O}%
_{\upsilon }(\varepsilon )$ vanishes when $\varepsilon \rightarrow 0^{+}$
for any fixed $\upsilon \in (0,\beta /2)$. By \cite[Eqs. (139) and (142)]%
{OhmIII},
\begin{equation}
\tilde{\Gamma}_{\upsilon ,\varepsilon ,1,1}(t)=\underset{l\rightarrow \infty
}{\lim }\frac{1}{\left\vert \Lambda _{l}\right\vert }\sum\limits_{x,y\in
\Lambda _{l}}\mathbb{E}\left[ \int\nolimits_{\upsilon }^{\beta -\upsilon }%
\mathfrak{B}_{t+i\alpha ,\upsilon ,\varepsilon }^{(\omega
)}(x,x-e_{1},y,y-e_{1})\mathrm{d}\alpha \right] <\infty
\label{gamma tilde cool}
\end{equation}%
for all $\varepsilon ,\beta \in \mathbb{R}^{+}$, $\lambda \in \mathbb{R}%
_{0}^{+}$, $t\in \mathbb{R}$ and $\upsilon \in (0,\beta /2)$, with%
\begin{multline*}
\mathfrak{B}_{t+i\alpha ,\upsilon ,\varepsilon }^{(\omega )}(\mathbf{x},%
\mathbf{y}):=\underset{\pi ,\pi ^{\prime }\in S_{2}}{\sum }\varepsilon _{\pi
}\varepsilon _{\pi ^{\prime }}B_{t+i\alpha ,\upsilon ,\varepsilon }^{(\omega
)}(y^{\pi ^{\prime }(1)},x^{\pi (1)}) \\
\times B_{-t+i(\beta -\alpha ),\upsilon ,\varepsilon }^{(\omega )}(x^{\pi
(2)},y^{\pi ^{\prime }(2)})
\end{multline*}%
and
\begin{equation}
B_{t+i\alpha ,\upsilon ,\varepsilon }^{(\omega )}\left( \mathbf{x}\right)
=\int\nolimits_{|\nu |<M_{\beta ,\upsilon ,\varepsilon }}\hat{F}_{\alpha
}^{\beta }\left( \nu \right) \langle \mathfrak{e}_{x^{(2)}},\mathrm{e}%
^{-i\left( t-\nu \right) \left( \Delta _{\mathrm{d}}+\lambda V_{\omega
}\right) }\mathfrak{e}_{x^{(1)}}\rangle \mathrm{d}\nu  \label{B rewritten}
\end{equation}%
for any $\mathbf{x}:=(x^{(1)},x^{(2)})\in \mathfrak{L}^{2}$ and $\mathbf{y}%
:=(y^{(1)},y^{(2)})\in \mathfrak{L}^{2}$. Here, $M_{\beta ,\upsilon
,\varepsilon }$ is a constant only depending on $\beta ,\upsilon
,\varepsilon $ and $\hat{F}_{\alpha }^{\beta }$ is the Fourier transform of
the function $F_{\alpha }^{\beta }$ (\ref{function F}). See \cite[Eq. (87)]%
{OhmIII}.

Thus, by (\ref{important estimate conduc}), it suffices to obtain the
asymptotics $\lambda \rightarrow 0^{+}$ and $\lambda \rightarrow \infty $ of
the function $\tilde{\Gamma}_{\upsilon ,\varepsilon ,1,1}$. To this end we
use the finite sum approximation%
\begin{multline*}
\xi _{\nu ,t,N}^{(\omega ,\lambda )}:=\mathrm{e}^{-i\left( t-\nu \right)
\lambda V_{\omega }}+\sum\limits_{n=1}^{N-1}(-i)^{n}\int_{\nu }^{t}\mathrm{d}%
\nu _{1}\int_{\nu }^{\nu _{1}}\mathrm{d}\nu _{2}\cdots \int_{\nu }^{\nu
_{n-1}}\mathrm{d}\nu _{n}\ \mathrm{e}^{-i\left( t-\nu _{1}\right) \lambda
V_{\omega }}\Delta _{\mathrm{d}} \\
\times \mathrm{e}^{-i\left( \nu _{1}-\nu _{2}\right) \lambda V_{\omega
}}\Delta _{\mathrm{d}}\mathrm{e}^{-i\left( \nu _{2}-\nu _{3}\right) \lambda
V_{\omega }}\cdots \mathrm{e}^{-i\left( \nu _{n-1}-\nu _{n}\right) \lambda
V_{\omega }}\Delta _{\mathrm{d}}\mathrm{e}^{-i\left( \nu _{n}-\nu \right)
\lambda V_{\omega }}\
\end{multline*}%
of the unitary operator $\mathrm{e}^{-i\left( t-\nu \right) \left( \Delta _{%
\mathrm{d}}+\lambda V_{\omega }\right) }$ for any $\omega \in \Omega $, $%
\lambda \in \mathbb{R}_{0}^{+}$, $N\in \mathbb{N}$ and $\nu ,t\in \mathbb{R}$%
. Indeed, using Duhamel's formula one gets that
\begin{equation}
\underset{N\rightarrow \infty }{\lim }\left\Vert \xi _{\nu ,t,N}^{(\omega
,\lambda )}-\mathrm{e}^{-i\left( t-\nu \right) \left( \Delta _{\mathrm{d}%
}+\lambda V_{\omega }\right) }\right\Vert _{\mathrm{op}}=0
\label{B rewrittenbis}
\end{equation}%
uniformly for $\omega \in \Omega $, $\lambda \in \mathbb{R}_{0}^{+}$, $\nu
\in \lbrack -M_{\beta ,\upsilon ,\varepsilon },M_{\beta ,\upsilon
,\varepsilon }]$ and times $t$ in compact sets. Hence, we replace $\mathrm{e}%
^{-i\left( t-\nu \right) \left( \Delta _{\mathrm{d}}+\lambda V_{\omega
}\right) }$ in (\ref{B rewritten}) by its approximation $\xi _{\nu
,t,N}^{(\omega ,\lambda )}$ and define%
\begin{equation}
\tilde{B}_{t+i\alpha ,\upsilon ,\varepsilon ,N}^{(\omega ,\lambda )}\left(
\mathbf{x}\right) :=\int\nolimits_{|\nu |<M_{\beta ,\upsilon ,\varepsilon }}%
\hat{F}_{\alpha }^{\beta }\left( \nu \right) \langle \mathfrak{e}%
_{x^{(2)}},\xi _{\nu ,t,N}^{(\omega ,\lambda )}\mathfrak{e}_{x^{(1)}}\rangle
\mathrm{d}\nu  \label{B rewrittenbisibis}
\end{equation}%
as well as
\begin{multline*}
\mathfrak{\tilde{B}}_{t+i\alpha ,\upsilon ,\varepsilon ,N}^{(\omega ,\lambda
)}(\mathbf{x},\mathbf{y}):=\underset{\pi ,\pi ^{\prime }\in S_{2}}{\sum }%
\varepsilon _{\pi }\varepsilon _{\pi ^{\prime }}\tilde{B}_{t+i\alpha
,\upsilon ,\varepsilon ,N}^{(\omega )}(y^{\pi ^{\prime }(1)},x^{\pi (1)}) \\
\times \tilde{B}_{-t+i(\beta -\alpha ),\upsilon ,\varepsilon ,N}^{(\omega
)}(x^{\pi (2)},y^{\pi ^{\prime }(2)})
\end{multline*}%
for any $\varepsilon ,\beta \in \mathbb{R}^{+}$, $\omega \in \Omega $, $%
\lambda \in \mathbb{R}_{0}^{+}$, $\mathbf{x}:=(x^{(1)},x^{(2)})\in \mathfrak{%
L}^{2}$ and $\mathbf{y}:=(y^{(1)},y^{(2)})\in \mathfrak{L}^{2}$. Indeed, one
has:

\begin{lemma}[Finite sum approximation]
\label{lemma limit cool}\mbox{
}\newline
Let $\varepsilon ,\beta \in \mathbb{R}^{+}$, $t\in \mathbb{R}$ and $\upsilon
\in (0,\beta /2)$. Then,%
\begin{multline*}
\underset{N\rightarrow \infty }{\lim }\frac{1}{\left\vert \Lambda
_{l}\right\vert }\sum\limits_{x,y\in \Lambda _{l}}\int\nolimits_{\upsilon
}^{\beta -\upsilon }\left\vert \mathfrak{B}_{t+i\alpha ,\upsilon
,\varepsilon }^{(\omega )}(x,x-e_{1},y,y-e_{1})\right. \\
\left. -\mathfrak{\tilde{B}}_{t+i\alpha ,\upsilon ,\varepsilon ,N}^{(\omega
,\lambda )}(x,x-e_{1},y,y-e_{1})\right\vert \mathrm{d}\alpha =0
\end{multline*}%
uniformly for $l\in \mathbb{R}^{+}$, $\omega \in \Omega $ and $\lambda \in
\mathbb{R}_{0}^{+}$.
\end{lemma}

\begin{proof}
The map $\left( \alpha ,\nu \right) \mapsto \hat{F}_{\alpha }^{\beta }\left(
\nu \right) $ is absolutely integrable in
\begin{equation*}
\left( \alpha ,\nu \right) \in \left[ \upsilon ,\beta -\upsilon \right]
\times \lbrack -M_{\beta ,\upsilon ,\varepsilon },M_{\beta ,\upsilon
,\varepsilon }]
\end{equation*}%
for any $\varepsilon ,\beta \in \mathbb{R}^{+}$ and $\upsilon \in (0,\beta
/2)$. Therefore, the assertion is directly proven by using (\ref{B
rewrittenbis}) to compute the difference between (\ref{B rewritten}) and (%
\ref{B rewrittenbisibis}). We omit the details. See similar arguments to the
proof of \cite[Lemma 5.11]{OhmIII}.
\end{proof}

As a consequence, we only need to bound, for any $\varepsilon ,\beta \in
\mathbb{R}^{+}$, $\upsilon \in (0,\beta /2)$, and $l,N\in \mathbb{N}$, the
function
\begin{equation*}
\mathfrak{q}_{\upsilon ,\varepsilon ,N,l}^{(\beta ,\omega ,\lambda )}\left(
t\right) :=\frac{1}{\left\vert \Lambda _{l}\right\vert }\sum\limits_{x,y\in
\Lambda _{l}}\mathbb{E}\left[ \int\nolimits_{\upsilon }^{\beta -\upsilon }%
\mathfrak{\tilde{B}}_{t+i\alpha ,\upsilon ,\varepsilon ,N}^{(\omega ,\lambda
)}(x,x-e_{1},y,y-e_{1})\mathrm{d}\alpha \right] ,
\end{equation*}%
as $\lambda \rightarrow 0^{+}$ and $\lambda \rightarrow \infty $, uniformly
for all $l\in \mathbb{R}^{+}$.

\begin{lemma}[Asymptotics of the finite sum approximation]
\label{lemma conductivty4 copy(7)}\mbox{
}\newline
Let $\varepsilon ,\beta \in \mathbb{R}^{+}$, $\lambda \in \mathbb{R}_{0}^{+}$%
, $t\in \mathbb{R}$, $\upsilon \in (0,\beta /2)$, and $N\in \mathbb{N}$.
Then,%
\begin{equation*}
\underset{\lambda \rightarrow 0}{\lim }\ \mathbb{E}\left[ \mathfrak{q}%
_{\upsilon ,\varepsilon ,N,l}^{(\beta ,\omega ,\lambda )}(t)\right] =\mathbb{%
E}\left[ \mathfrak{q}_{\upsilon ,\varepsilon ,N,l}^{(\beta ,\omega ,0)}(t)%
\right]
\end{equation*}%
uniformly for $l\in \mathbb{R}^{+}$. If the probability measure $\mathfrak{a}%
_{\mathbf{0}}$ is in addition absolutely continuous w.r.t. the Lebesgue
measure then
\begin{equation*}
\underset{\lambda \rightarrow \infty }{\lim }\mathbb{E}\left[ \mathfrak{q}%
_{\upsilon ,\varepsilon ,N,l}^{(\beta ,\omega ,\lambda )}(t)\right] =0
\end{equation*}%
uniformly for $l\in \mathbb{R}^{+}$.
\end{lemma}

\begin{proof}
The function $\mathfrak{q}_{\upsilon ,\varepsilon ,N,l}^{(\beta ,\omega
,\lambda )}(t)$ is a finite sum of terms of the form
\begin{eqnarray*}
&&\frac{(-i)^{n_{1}+n_{2}}}{\left\vert \Lambda _{l}\right\vert }%
\sum\limits_{x,y\in \Lambda _{l}}\underset{\pi ,\pi ^{\prime }\in S_{2}}{%
\sum }\varepsilon _{\pi }\varepsilon _{\pi ^{\prime
}}\int\nolimits_{\upsilon }^{\beta -\upsilon }\mathrm{d}\alpha
\int\nolimits_{|\nu |<M_{\beta ,\upsilon ,\varepsilon }}\mathrm{d}\nu
\int\nolimits_{|u|<M_{\beta ,\upsilon ,\varepsilon }}\mathrm{d}u \\
&&\int_{\nu }^{t}\mathrm{d}\nu _{1}\cdots \int_{\nu }^{\nu _{n_{1}-1}}%
\mathrm{d}\nu _{n_{1}}\int_{u}^{-t}\mathrm{d}u_{1}\cdots
\int_{u}^{u_{n_{2}-1}}\mathrm{d}u_{n_{2}}\ \hat{F}_{\alpha }^{\beta }\left(
\nu \right) \hat{F}_{\beta -\alpha }^{\beta }\left( u\right) \\
&&\times \langle \mathfrak{e}_{x_{\pi (1)}},\mathrm{e}^{-i(t-\nu
_{1})\lambda V_{\omega }}\Delta _{\mathrm{d}}\mathrm{e}^{-i(\nu _{1}-\nu
_{2})\lambda V_{\omega }}\Delta _{\mathrm{d}}\cdots \\
&&\qquad \qquad \qquad \qquad \qquad \qquad \quad \cdots \mathrm{e}%
^{-i\left( \nu _{n_{1}-1}-\nu _{n_{1}}\right) \lambda V_{\omega }}\Delta _{%
\mathrm{d}}\mathrm{e}^{-i(\nu _{n_{1}}-\nu )\lambda V_{\omega }}\mathfrak{e}%
_{y_{\pi ^{\prime }(1)}}\rangle \\
&&\times \langle \mathfrak{e}_{y_{\pi ^{\prime }(2)}},\mathrm{e}%
^{-i(-t-u_{1})\lambda V_{\omega }}\Delta _{\mathrm{d}}\mathrm{e}%
^{-i(u_{1}-u_{2})\lambda V_{\omega }}\Delta _{\mathrm{d}}\cdots \\
&&\qquad \qquad \qquad \qquad \qquad \qquad \qquad \cdots \mathrm{e}%
^{i\left( u_{n_{2}-1}-u_{n_{2}}\right) \lambda V_{\omega }}\Delta _{\mathrm{d%
}}\mathrm{e}^{-i(u_{n_{2}}-u)\lambda V_{\omega }}\mathfrak{e}_{x_{\pi
(2)}}\rangle
\end{eqnarray*}%
for $n_{1},n_{2}\in \mathbb{N}_{0}\cap \left[ 0,N\right] $. Here, $%
(x_{1},x_{2}):=(x,x-e_{1})$, $(y_{1},y_{2}):=(y,y-e_{1})$. [By abuse of
notation, the case $n_{1}=0$ or $n_{2}=0$ means that there is no integral
but a term $\mathrm{e}^{-i(t-\nu )\lambda V_{\omega }}$ inside the
corresponding scalar product.] From this and the translation invariance of
the probability measure $\mathfrak{a}_{\Omega }$, we get that $\mathbb{E}[%
\mathfrak{q}_{\upsilon ,\varepsilon ,N,l}^{(\beta ,\omega ,\lambda )}(t)]$
is a finite sum of terms of the form
\begin{eqnarray}
&&(-i)^{n_{1}+n_{2}}\sum\limits_{x\in \mathfrak{L}}\underset{\pi ,\pi
^{\prime }\in S_{2}}{\sum }\varepsilon _{\pi }\varepsilon _{\pi ^{\prime
}}\int\nolimits_{\upsilon }^{\beta -\upsilon }\mathrm{d}\alpha
\int\nolimits_{|\nu |<M_{\beta ,\upsilon ,\varepsilon }}\mathrm{d}\nu
\int\nolimits_{|u|<M_{\beta ,\upsilon ,\varepsilon }}\mathrm{d}u
\label{ohm idiot1} \\
&&\int_{\nu }^{t}\mathrm{d}\nu _{1}\cdots \int_{\nu }^{\nu _{n_{1}-1}}%
\mathrm{d}\nu _{n_{1}}\int_{u}^{-t}\mathrm{d}u_{1}\cdots
\int_{u}^{u_{n_{2}-1}}\mathrm{d}u_{n_{2}}\hat{F}_{\alpha }^{\beta }\left(
\nu \right) \hat{F}_{\beta -\alpha }^{\beta }\left( u\right)  \notag \\
&& \sum\limits_{z \in \Lambda_l} \frac{\mathbf{1}[x+z \in \Lambda _{l}]}{%
|\Lambda_l |}\mathbb{E}%
\Big[%
\langle \mathfrak{e}_{x_{\pi (1)}},\mathrm{e}^{-i(t-\nu _{1})\lambda
V_{\omega }}\Delta _{\mathrm{d}}\mathrm{e}^{-i(\nu _{1}-\nu _{2})\lambda
V_{\omega }}\Delta _{\mathrm{d}}\cdots  \notag \\
&&\qquad \qquad \qquad \qquad \qquad \qquad \quad \cdots \mathrm{e}%
^{-i\left( \nu _{n_{1}-1}-\nu _{n_{1}}\right) \lambda V_{\omega }}\Delta _{%
\mathrm{d}}\mathrm{e}^{-i(\nu _{n_{1}}-\nu )\lambda V_{\omega }}\mathfrak{e}%
_{y_{\pi ^{\prime }(1)}}\rangle  \notag \\
&&\times \langle \mathfrak{e}_{y_{\pi ^{\prime }(2)}},\mathrm{e}%
^{-i(-t-u_{1})\lambda V_{\omega }}\Delta _{\mathrm{d}}\mathrm{e}%
^{-i(u_{1}-u_{2})\lambda V_{\omega }}\Delta _{\mathrm{d}}\cdots  \notag \\
&&\qquad \qquad \qquad \qquad \qquad \qquad \quad \cdots \mathrm{e}%
^{-i\left( u_{n_{2}-1}-u_{n_{2}}\right) \lambda V_{\omega }}\Delta _{\mathrm{%
d}}\mathrm{e}^{-i(u_{n_{2}}-u)\lambda V_{\omega }}\mathfrak{e}_{x_{\pi
(2)}}\rangle
\Big]%
\text{ },  \notag
\end{eqnarray}%
where $(x_{1},x_{2}):=(x,x-e_{1})$, $(y_{1},y_{2}):=(0,-e_{1})$. Note that
\begin{equation*}
\int\nolimits_{\upsilon }^{\beta -\upsilon }\mathrm{d}\alpha
\int\nolimits_{|\nu |<M_{\beta ,\upsilon ,\varepsilon }}\mathrm{d}\nu
\int\nolimits_{|u|<M_{\beta ,\upsilon ,\varepsilon }}\mathrm{d}u\left\vert
\hat{F}_{\alpha }^{\beta }\left( \nu \right) \hat{F}_{\beta -\alpha }^{\beta
}\left( u\right) \right\vert <\infty
\end{equation*}%
and the volume of integration in (\ref{ohm idiot1}) of the $\nu _{a}$-- and $%
u_{b}$--integrals, $a=1,\ldots ,n_{1}$, $b=1,\ldots ,n_{2}$, gives a factor
\begin{equation*}
\frac{|t-\nu |^{n_{1}}|t+u|^{n_{2}}}{n_{1}!n_{2}!}.
\end{equation*}%
By developing the Laplacians $\Delta _{\mathrm{d}}$, note that, whenever $%
t\neq \nu $, $t\neq -u$,
\begin{eqnarray*}
&&\sum\limits_{z \in \Lambda_l} \frac{\mathbf{1}[x+z \in \Lambda _{l}]}{%
|\Lambda_l |} \mathbb{E}%
\Big[%
\langle \mathfrak{e}_{x_{\pi (1)}},\mathrm{e}^{-i(t-\nu _{1})\lambda
V_{\omega }}\Delta _{\mathrm{d}}\mathrm{e}^{-i(\nu _{1}-\nu _{2})\lambda
V_{\omega }}\Delta _{\mathrm{d}} \\
&&\qquad \qquad \qquad \qquad \qquad \quad \cdots \mathrm{e}^{-i\left( \nu
_{n_{1}-1}-\nu _{n_{1}}\right) \lambda V_{\omega }}\Delta _{\mathrm{d}}%
\mathrm{e}^{-i(\nu _{n_{1}}-\nu )\lambda V_{\omega }}\mathfrak{e}_{y_{\pi
^{\prime }(1)}}\rangle \\
&&\times \langle \mathfrak{e}_{y_{\pi ^{\prime }(2)}},\mathrm{e}%
^{-i(-t-u_{1})\lambda V_{\omega }}\Delta _{\mathrm{d}}\mathrm{e}%
^{-i(u_{1}-u_{2})\lambda V_{\omega }}\Delta _{\mathrm{d}} \\
&&\qquad \qquad \qquad \qquad \qquad \quad \cdots \mathrm{e}^{-i\left(
u_{n_{2}-1}-u_{n_{2}}\right) \lambda V_{\omega }}\Delta _{\mathrm{d}}\mathrm{%
e}^{-i(u_{n_{2}}-u)\lambda V_{\omega }}\mathfrak{e}_{x_{\pi (2)}}\rangle
\Big]%
\end{eqnarray*}%
is a sum of $(2d+1)^{n_{1}+n_{2}}$ terms of the form, up to constants
bounded in absolute value by $(2d)^{n_{1}+n_{2}}$,
\begin{equation}
\sum\limits_{z \in \Lambda_l} \frac{\mathbf{1}[x+z \in \Lambda _{l}]}{%
|\Lambda_l |} \mathbf{1}[x\in \Lambda _{2N+1}]\mathbb{E}\left[ \mathrm{e}%
^{\pm i\mathfrak{t}_{1}\lambda V_{\omega }\left( x_{1}\right) }\cdots
\mathrm{e}^{\pm i\mathfrak{t}_{n}\lambda V_{\omega }\left( x_{n}\right) }%
\right]  \label{blqblq}
\end{equation}%
where $n\in \mathbb{N}$, $n\leq n_{1}+n_{2}\leq 2N$, $\mathfrak{t}%
_{1},\ldots ,\mathfrak{t}_{n}\in \mathbb{R}^{+}$ and $x_{1}\in \{x,x-e_{1}\}$%
, $x_{2}\ldots ,x_{n-1}\in \mathfrak{L}$, $x_{n}\in \{0,-e_{1}\}$ with $%
x_{j}\neq x_{p}$ for $j\neq p$. By Lebesgue's dominated convergence theorem,
it suffices to analyze (\ref{blqblq}) either in the limit $\lambda
\rightarrow \infty $ or $\lambda \rightarrow 0^{+}$. By (\ref{probability
measure}),
\begin{equation}
\mathbb{E}\left[ \mathrm{e}^{\pm i\mathfrak{t}_{1}\lambda V_{\omega }\left(
x_{1}\right) }\cdots \mathrm{e}^{\pm i\mathfrak{t}_{n}\lambda V_{\omega
}\left( x_{n}\right) }\right] =\mathbb{E}\left[ \mathrm{e}^{\pm i\mathfrak{t}%
_{1}\lambda V_{\omega }\left( x_{1}\right) }\right] \cdots \mathbb{E}\left[
\mathrm{e}^{\pm i\mathfrak{t}_{n}\lambda V_{\omega }\left( x_{n}\right) }%
\right]  \label{blqblqbis}
\end{equation}%
for any $n\in \mathbb{N}$, $\mathfrak{t}_{1},\ldots ,\mathfrak{t}_{n}\in
\mathbb{R}^{+}$ and $x_{1},\ldots ,x_{n}\in \mathfrak{L}$ with $x_{j}\neq
x_{p}$ for $j\neq p$. Since
\begin{equation*}
\underset{\lambda \rightarrow 0}{\lim }\mathbb{E}\left[ \mathrm{e}^{\pm i%
\mathfrak{t}\lambda V_{\omega }\left( x\right) }\right] =1
\end{equation*}%
for all $x\in \mathfrak{L}$ and $\mathfrak{t}\in \mathbb{R}^{+}$, we deduce
from (\ref{blqblqbis}) that
\begin{equation*}
\underset{\lambda \rightarrow 0}{\lim }\mathbb{E}\left[ \mathrm{e}^{\pm i%
\mathfrak{t}_{1}\lambda V_{\omega }\left( x_{1}\right) }\cdots \mathrm{e}%
^{\pm i\mathfrak{t}_{n}\lambda V_{\omega }\left( x_{n}\right) }\right] =1
\end{equation*}%
and one gets the first assertion of the lemma by Lebesgue's dominated
convergence theorem.

If, additionally, the probability measure $\mathfrak{a}_{\mathbf{0}}$ is a
absolutely continuous w.r.t. the Lebesgue measure, then from the
Riemann--Lebesgue lemma we have the limit
\begin{equation*}
\underset{\lambda \rightarrow \infty }{\lim }\mathbb{E}\left[ \mathrm{e}%
^{\pm i\mathfrak{t}\lambda V_{\omega }\left( x\right) }\right] =0
\end{equation*}%
for all $x\in \mathfrak{L}$ and $\mathfrak{t}\in \mathbb{R}^{+}$. From (\ref%
{blqblqbis}), we then obtain that
\begin{equation*}
\underset{\lambda \rightarrow \infty }{\lim }\mathbb{E}\left[ \mathrm{e}%
^{\pm i\mathfrak{t}_{1}\lambda V_{\omega }\left( x_{1}\right) }\cdots
\mathrm{e}^{\pm i\mathfrak{t}_{n}\lambda V_{\omega }\left( x_{n}\right) }%
\right] =0 \ .
\end{equation*}%
Using this and Lebesgue's dominated convergence theorem, one thus gets the
second assertion.
\end{proof}

We are now in position to compute the asymptotics, as $\lambda \rightarrow
0^{+}$ and $\lambda \rightarrow \infty $, of the paramagnetic conductivity $%
\mathbf{\sigma }_{\mathrm{p}}$, which equals (\ref{important estimate conduc}%
).

\begin{proposition}[Asymptotics of the paramagnetic conductivity]
\label{main 1 copy(5)}\mbox{
}\newline
Let $\beta \in \mathbb{R}^{+}$, $\lambda \in \mathbb{R}_{0}^{+}$ and $t\in
\mathbb{R}$. Then,%
\begin{equation*}
\underset{\lambda \rightarrow 0}{\lim }\mathbf{\sigma }_{\mathrm{p}}^{(\beta
,\lambda )}(t)=\mathbf{\sigma }_{\mathrm{p}}^{(\beta ,0)}(t)\ .
\end{equation*}%
If the probability measure $\mathfrak{a}_{\mathbf{0}}$ is in addition
absolutely continuous w.r.t. the Lebesgue measure then
\begin{equation*}
\underset{\lambda \rightarrow \infty }{\lim }\mathbf{\sigma }_{\mathrm{p}%
}^{(\beta ,\lambda )}(t)=0\ .
\end{equation*}
\end{proposition}

\begin{proof}
Let $\beta \in \mathbb{R}^{+}$, $\lambda \in \mathbb{R}_{0}^{+}$ and $t\in
\mathbb{R}$. By Lemmata \ref{lemma limit cool}--\ref{lemma conductivty4
copy(7)},%
\begin{eqnarray*}
&&\underset{\lambda \rightarrow 0}{\lim }\frac{1}{\left\vert \Lambda
_{l}\right\vert }\sum\limits_{x,y\in \Lambda _{l}}\mathbb{E}\left[
\int\nolimits_{\upsilon }^{\beta -\upsilon }\mathfrak{B}_{t+i\alpha
,\upsilon ,\varepsilon }^{(\beta ,\omega ,\lambda )}(x,x-e_{1},y,y-e_{1})%
\mathrm{d}\alpha \right] \\
&=&\frac{1}{\left\vert \Lambda _{l}\right\vert }\sum\limits_{x,y\in \Lambda
_{l}}\mathbb{E}\left[ \int\nolimits_{\upsilon }^{\beta -\upsilon }\mathfrak{B%
}_{t+i\alpha ,\upsilon ,\varepsilon }^{(\beta ,\omega
,0)}(x,x-e_{1},y,y-e_{1})\mathrm{d}\alpha \right]
\end{eqnarray*}%
uniformly for all $l\in \mathbb{R}^{+}$, whereas
\begin{equation*}
\underset{\lambda \rightarrow \infty }{\lim }\frac{1}{\left\vert \Lambda
_{l}\right\vert }\sum\limits_{x,y\in \Lambda _{l}}\mathbb{E}\left[
\int\nolimits_{\upsilon }^{\beta -\upsilon }\mathfrak{B}_{t+i\alpha
,\upsilon ,\varepsilon }^{(\beta ,\omega ,\lambda )}(x,x-e_{1},y,y-e_{1})%
\mathrm{d}\alpha \right] =0
\end{equation*}%
uniformly for all $l\in \mathbb{R}^{+}$, provided the probability measure $%
\mathfrak{a}_{\mathbf{0}}$ is absolutely continuous w.r.t. the Lebesgue
measure. Thus, by using these limits together with (\ref{important estimate
conduc})--(\ref{gamma tilde cool}) we arrive at the assertions.
\end{proof}

Finally, to get Theorem \ref{main 4}, we need to compute explicitly the
paramagnetic conductivity $\mathbf{\sigma }_{\mathrm{p}}^{(\beta ,\lambda )}$
at $\lambda =0$. This is done in the next lemma:

\begin{lemma}[Paramagnetic conductivity at constant potential]
\label{lemma conductivty4 copy(1)}\mbox{
}\newline
For any $\beta \in \mathbb{R}^{+}\ $and $t\in \mathbb{R}$, $\mathbf{\sigma }%
_{\mathrm{p}}^{(\beta ,0)}(t)=0$.
\end{lemma}

\begin{proof}
Let $\beta \in \mathbb{R}^{+}$. By (\ref{paramagnetic transport coefficient
macro}) and \cite[Lemma 5.2]{OhmII}, note that%
\begin{equation}
\mathbf{\sigma }_{\mathrm{p}}^{(\beta ,0)}(t)=\underset{l\rightarrow \infty }%
{\lim }\frac{1}{\left\vert \Lambda _{l}\right\vert }\underset{x,y\in \Lambda
_{l}}{\sum }\int\nolimits_{0}^{\beta }\left( \mathfrak{D}_{t+i\alpha }(x,y)-%
\mathfrak{D}_{i\alpha }(x,y)\right) \mathrm{d}\alpha \ ,  \label{idiot 00}
\end{equation}%
where, for any $x,y\in \mathfrak{L}$,%
\begin{equation*}
\mathfrak{D}_{t+i\alpha }(x,y):=\mathfrak{C}_{t+i\alpha }^{(\beta ,\omega
,0)}(x,x-e_{1},y,y-e_{1})\ .
\end{equation*}%
Observe also that $\mathfrak{C}_{t+i\alpha }^{(\beta ,\omega ,0)}$, which is
defined by (\ref{map coolbis}), does not depend on $\omega \in \Omega $.
Explicit computations show that $\mathfrak{D}_{t+i\alpha }(x,y)$ equals%
\begin{eqnarray*}
\mathfrak{D}_{t+i\alpha }(x,y) &=&\frac{2}{(2\pi )^{2d}}\int\nolimits_{\left[
-\pi ,\pi \right] ^{d}}\mathrm{d}^{d}p\int\nolimits_{\left[ -\pi ,\pi \right]
^{d}}\mathrm{d}^{d}p^{\prime }\ \frac{\mathrm{e}^{\beta E(p^{\prime })}%
\mathrm{e}^{\left( \alpha -it\right) \left( E(p)-E(p^{\prime })\right) }}{%
\left( 1+\mathrm{e}^{\beta E(p)}\right) \left( 1+\mathrm{e}^{\beta
E(p^{\prime })}\right) } \\
&&\times \left( 1-\cos \left( p_{1}-p_{1}^{\prime }\right) \right) \mathrm{e}%
^{i\left( p+p^{\prime }\right) \cdot \left( x-y\right) }
\end{eqnarray*}%
for any $t\in \mathbb{R}$, $\alpha \in \lbrack 0,\beta ]$ and $x,y\in
\mathfrak{L}$, with $E\left( p\right) =E\left( -p\right) $ being the
dispersion relation (\ref{dispertion relation section2}) of $\Delta _{%
\mathrm{d}}$. By performing the transformation $p\rightarrow p-p^{\prime }$
and then $p^{\prime }\rightarrow p^{\prime }+p/2$ together with $E\left(
p\right) =E\left( -p\right) $ we deduce that
\begin{equation}
\int\nolimits_{0}^{\beta }\mathfrak{D}_{t+i\alpha }(x,y)\mathrm{d}\alpha
=\int\nolimits_{\left[ -\pi ,\pi \right] ^{d}}\mathfrak{d}_{t}\left(
p\right) \mathrm{e}^{ip\cdot \left( x-y\right) }\mathrm{d}^{d}p
\label{eq idiot sup}
\end{equation}%
for all $t\in \mathbb{R}$ and $x,y\in \mathfrak{L}$, with $\mathfrak{d}_{t}$
being the function defined on $\left[ -\pi ,\pi \right] ^{d}$ by%
\begin{eqnarray*}
\mathfrak{d}_{t}\left( p\right)  &:=&\frac{2}{(2\pi )^{2d}}\int\nolimits_{\left[ -\pi ,\pi \right] ^{d}}\mathrm{d}^{d}p^{\prime }\ \frac{\mathrm{e}^{\beta E(p^{\prime }+p/2)}\mathrm{e}^{-it\left( E(p^{\prime
}-p/2)-E(p^{\prime }+p/2)\right) }}{\left( 1+\mathrm{e}^{\beta E(p^{\prime
}-p/2)}\right) \left( 1+\mathrm{e}^{\beta E(p^{\prime }+p/2)}\right) } \\
&&\times \frac{\left( \mathrm{e}^{\beta \left( E(p^{\prime
}-p/2)-E(p^{\prime }+p/2)\right) }-1\right) }{\left( E(p^{\prime
}-p/2)-E(p^{\prime }+p/2)\right) }\left( 1-\cos \left( 2p_{1}^{\prime
}\right) \right) \ .
\end{eqnarray*}%
Consequently, using (\ref{eq idiot sup}) one gets, for any $l\in \mathbb{R}%
^{+}$ and $t\in \mathbb{R}$, the equality
\begin{equation}
\frac{1}{\left\vert \Lambda _{l}\right\vert }\sum\limits_{x,y\in \Lambda
_{l}}\int\nolimits_{0}^{\beta }\mathfrak{D}_{t+i\alpha }(x,y)\mathrm{d}%
\alpha =\int\nolimits_{\left[ -\pi ,\pi \right] ^{d}}\gamma _{l}\left(
p\right) \mathfrak{d}_{t}\left( p\right) \mathrm{d}^{d}p\ ,
\label{idiot 000}
\end{equation}%
where the function $\gamma _{l}$ is defined on $\left[ -\pi ,\pi \right]
^{d} $ by
\begin{equation*}
\gamma _{l}\left( p\right) :=\left\vert \frac{1}{\left\vert \Lambda
_{l}\right\vert ^{1/2}}\sum\limits_{x\in \Lambda _{l}}\mathrm{e}^{ip\cdot
x}\right\vert ^{2}=\frac{1}{\left\vert \Lambda _{l}\right\vert }%
\sum\limits_{x,y\in \Lambda _{l}}\mathrm{e}^{ip\cdot \left( x-y\right) }\ .
\end{equation*}
Observe that, for any $l\in \mathbb{R}^{+}$ and all $\varepsilon \in \mathbb{%
R}^{+}$,
\begin{equation*}
\int\nolimits_{\left[ -\pi ,\pi \right] ^{d}}\gamma _{l}\left( p\right)
\mathrm{d}^{d}p=(2\pi )^{2d}\qquad \text{and}\qquad \underset{l\rightarrow
\infty }{\lim }\int\nolimits_{\left[ -\pi ,\pi \right] ^{d}\backslash
\mathcal{B}\left( 0,\varepsilon \right) }\gamma _{l}\left( p\right) \mathrm{d%
}^{d}p=0\ ,
\end{equation*}%
where $\mathcal{B}\left( 0,\varepsilon \right) \subset \mathbb{R}^{d}$ is
the ball of radius $\varepsilon $ centered at $0$. From this we infer that
\begin{equation}
\underset{l\rightarrow \infty }{\lim }\left\vert \int\nolimits_{\left[ -\pi
,\pi \right] ^{d}}\gamma _{l}\left( p\right) \mathfrak{d}_{t}\left( p\right)
\mathrm{d}^{d}p-\int\nolimits_{\mathcal{B}\left( 0,\varepsilon \right)
}\gamma _{l}\left( p\right) \mathfrak{d}_{t}\left( p\right) \mathrm{d}%
^{d}p\right\vert =0  \label{idiot encore}
\end{equation}%
for all $\varepsilon \in \mathbb{R}^{+}$ and any $t\in \mathbb{R}$.
Meanwhile, remark that
\begin{equation*}
\mathfrak{d}_{t}\left( p\right) -\mathfrak{d}_{0}\left( p\right) =\mathcal{O}%
\left( |tp|\right) \ .
\end{equation*}%
Then, using the continuity of the function $\mathfrak{d}_{0}\left( \cdot
\right) $ together with (\ref{idiot 00}), (\ref{idiot 000}) and (\ref{idiot
encore}), it follows that $\mathbf{\sigma }_{\mathrm{p}}^{(\beta ,0)}(t)=0$
for all $t\in \mathbb{R}$.
\end{proof}

Therefore, Theorem \ref{main 4} (p) follows from Proposition \ref{main 1
copy(5)} and Lemma \ref{lemma conductivty4 copy(1)}.

\subsection{On the Strict Positivity of the Heat Production\label{Section
Strict Positivity}}

In this subsection we aim to prove Theorem \ref{main 2 copy(1)}: First, we
study the asymptotics of the paramagnetic conductivity $\mathbf{\sigma }_{%
\mathrm{p}}$ at $\beta ,\lambda ,t=0$. Then, we show that the behavior of $%
\mathbf{\sigma }_{\mathrm{p}}$ near this point implies strict positivity of
the heat production, at least for short pulses of the electric field and
small $\beta ,\lambda >0$. This result corresponds to Lemma \ref{main 2
copy(2)}. The latter can be extended at small $\beta ,\lambda >0$ by an
analyticity argument to all electric fields outside a meager set, see Lemma %
\ref{lemma conductivty4 copy(7)2}.

\begin{lemma}[Non--vanishing AC--conductivity measure -- I]
\label{main 2 copy(2)}\mbox{
}\newline
Let $\mathbf{A}\in C_{0}^{\infty }(\mathbb{R}\times {\mathbb{R}}^{d};{%
\mathbb{R}}^{d})\backslash \{0\}$ be such that, for some $k\in \{1,\ldots
,d\}$,
\begin{equation*}
\int\nolimits_{\mathbb{R}^{d}}\left( \int\nolimits_{\mathbb{R}}s\{E_{\mathbf{%
A}}(s,x)\}_{k}\mathrm{d}s\right) ^{2}\mathrm{d}^{d}x>0
\end{equation*}%
and define, for all $T\in \mathbb{R}^{+}$, the time--rescaled potential%
\begin{equation*}
\mathbf{A}^{(T)}(t,x):=\mathbf{A}(T^{-1}t,x)\ ,\quad t\in \mathbb{R},\ x\in
\mathbb{R}^{d}\ .
\end{equation*}%
For any $\lambda _{0}\in \mathbb{R}^{+}$, there are $\beta _{0},T_{0}\in
\mathbb{R}^{+}$ such that, for $\beta \in (0,\beta _{0})$, $\lambda \in
(\lambda _{0}/2,\lambda _{0})$ and $T\in (T_{0}/2,T_{0})$,
\begin{equation*}
\int\nolimits_{\mathbb{R}}\mathrm{d}s_{1}\int\nolimits_{\mathbb{R}}\mathrm{d}%
s_{2}\ \mathbf{\sigma }_{\mathrm{p}}(s_{1}-s_{2})\int\nolimits_{\mathbb{R}%
^{d}}\mathrm{d}^{d}x\left\langle E_{\mathbf{A}^{(T)}}(s_{2},x),E_{\mathbf{A}%
^{(T)}}(s_{1},x)\right\rangle >0\ .
\end{equation*}
\end{lemma}

\begin{proof}
Let $\lambda _{0}\in \mathbb{R}^{+}$. Using Duhamel's formula note first that%
\begin{equation}
\mathrm{e}^{\left( \alpha -it\right) \left( \Delta _{\mathrm{d}}+\lambda
V_{\omega }\right) }=\mathrm{e}^{\left( \alpha -it\right) \Delta _{\mathrm{d}%
}}+\int_{0}^{1}\mathrm{e}^{\left( \alpha -it\right) (1-\gamma )\Delta _{%
\mathrm{d}}}\left( \alpha -it\right) \lambda V_{\omega }\mathrm{e}^{\left(
\alpha -it\right) \gamma \left( \Delta _{\mathrm{d}}+\lambda V_{\omega
}\right) }\mathrm{d}\gamma  \label{Duhamel's formula encore}
\end{equation}%
for any $\alpha \in \lbrack 0,\beta ]$ and $t\in {\mathbb{R}}$. Since all
operators in this last equation are bounded, it follows that, if $\lambda
\in \lbrack 0,\lambda _{0}]$ and $\beta \in \mathbb{R}^{+}$ is sufficiently
small, the Neumann series for $\left( 1+\mathrm{e}^{\beta \left( \Delta _{%
\mathrm{d}}+\lambda V_{\omega }\right) }\right) ^{-1}$ absolutely converges:%
\begin{eqnarray}
&&\left( 1+\mathrm{e}^{\beta \left( \Delta _{\mathrm{d}}+\lambda V_{\omega
}\right) }\right) ^{-1}  \label{von neuman series} \\
&=&\underset{n=0}{\overset{\infty }{\sum }}\left\{ -\beta \lambda \left( 1+%
\mathrm{e}^{\beta \Delta _{\mathrm{d}}}\right) ^{-1}\int_{0}^{1}\mathrm{e}%
^{\beta (1-\gamma )\Delta _{\mathrm{d}}}V_{\omega }\mathrm{e}^{\beta \gamma
\left( \Delta _{\mathrm{d}}+\lambda V_{\omega }\right) }\mathrm{d}\gamma
\right\} ^{n}\left( 1+\mathrm{e}^{\beta \Delta _{\mathrm{d}}}\right) ^{-1}\ .
\notag
\end{eqnarray}%
By (\ref{Duhamel's formula encore})--(\ref{von neuman series}), one gets the
existence of a constant $D\in \mathbb{R}^{+}$ such that, for $\lambda \in
\lbrack 0,\lambda _{0}]$ and any sufficiently small $\beta \in \left(
0,1\right) $, $\alpha \in \lbrack 0,\beta ]$ and $\omega \in \Omega $,%
\begin{equation}
\left\Vert F_{\alpha }^{\beta }\left( \Delta _{\mathrm{d}}+\lambda V_{\omega
}\right) -F_{\alpha }^{\beta }\left( \Delta _{\mathrm{d}}\right) \right\Vert
_{\mathrm{op}}\leq D\beta \lambda  \label{positivity 1}
\end{equation}%
with $F_{\alpha }^{\beta }$ defined by (\ref{function F}).

We define the approximated complex--time two--point correlation function $%
\tilde{C}_{t+i\alpha }^{(\omega )}$, for any $\beta \in \mathbb{R}^{+}$, $%
\omega \in \Omega $, $\lambda \in \mathbb{R}_{0}^{+}$, $t\in {\mathbb{R}}$
and $\alpha \in \lbrack 0,\beta ]$, by%
\begin{equation}
\tilde{C}_{t+i\alpha }^{(\omega )}(\mathbf{x}):=\langle \mathfrak{e}%
_{x^{(2)}},\mathrm{e}^{-it\left( \Delta _{\mathrm{d}}+\lambda V_{\omega
}\right) }F_{\alpha }^{\beta }\left( \Delta _{\mathrm{d}}\right) \mathfrak{e}%
_{x^{(1)}}\rangle \ ,\quad \mathbf{x}:=(x^{(1)},x^{(2)})\in \mathfrak{L}%
^{2}\ ,  \label{positivity eq30}
\end{equation}%
compare with (\ref{cond two--point correlation function}), the original form
of $C_{t+i\alpha }^{(\omega )}$. For any $\mathbf{x}:=(x^{(1)},x^{(2)})\in
\mathfrak{L}^{2}$ and $\mathbf{y}:=(y^{(1)},y^{(2)})\in \mathfrak{L}^{2}$,
let us define
\begin{equation*}
\mathfrak{\tilde{C}}_{t+i\alpha }^{(\omega )}(\mathbf{x},\mathbf{y}):=%
\underset{\pi ,\pi ^{\prime }\in S_{2}}{\sum }\varepsilon _{\pi }\varepsilon
_{\pi ^{\prime }}\tilde{C}_{t+i\alpha }^{(\omega )}(y^{\pi ^{\prime
}(1)},x^{\pi (1)})\tilde{C}_{-t+i(\beta -\alpha )}^{(\omega )}(x^{\pi
(2)},y^{\pi ^{\prime }(2)})\ .
\end{equation*}%
From (\ref{map coolbis})--(\ref{positivity eq2}) and (\ref{positivity 1}) we
thus deduce that%
\begin{equation}
\Gamma _{1,1}(t)=\underset{l\rightarrow \infty }{\lim }\frac{1}{\left\vert
\Lambda _{l}\right\vert }\sum\limits_{x,y\in \Lambda _{l}}\mathbb{E}\left[
\int\nolimits_{0}^{\beta }\mathfrak{\tilde{C}}_{t+i\alpha }^{(\omega
)}(x,x-e_{1},y,y-e_{1})\mathrm{d}\alpha \right] +\mathcal{O}(\beta
^{2}\lambda )  \label{positivity eq3}
\end{equation}%
uniformly for $t\in {\mathbb{R}}$.

Next, we define an approximation of $\tilde{C}_{t+i\alpha }^{(\omega )}$ by%
\begin{eqnarray}
\hat{C}_{t+i\alpha }^{(\omega )}(\mathbf{x})&:=&\left\langle \mathfrak{e}_{x^{(2)}},\mathrm{e}^{-it\Delta _{\mathrm{d}}}F_{\alpha }^{\beta }\left(
\Delta _{\mathrm{d}}\right) \mathfrak{e}_{x^{(1)}}\right\rangle
\label{positivity eq3bis0} \\
&&-\frac{\lambda }{2}\left\langle \mathfrak{e}_{x^{(2)}},\left( itV_{\omega
}+\frac{t^{2}}{2}(V_{\omega }\Delta _{\mathrm{d}}+\Delta _{\mathrm{d}}V_{\omega }+\lambda V_{\omega }^{2})\right) \mathfrak{e}_{x^{(1)}}\right\rangle   \notag
\end{eqnarray}%
for all $\beta \in \mathbb{R}^{+}$, $\omega \in \Omega $, $\lambda \in
\mathbb{R}_{0}^{+}$, $t\in {\mathbb{R}}$, $\alpha \in \lbrack 0,\beta ]$ and
$\mathbf{x}:=(x^{(1)},x^{(2)})\in \mathfrak{L}^{2}$. Indeed, by (\ref%
{Duhamel's formula encore}) and a power expansion of $F_{\alpha }^{\beta
}\left( \Delta _{\mathrm{d}}\right) $ at $\alpha ,\beta =0$, there is a
constant $D\in \mathbb{R}^{+}$ such that, for any $\lambda \in \lbrack
0,\lambda _{0}]$, sufficiently small $\beta \in \left( 0,1\right) $, $\alpha
\in \lbrack 0,\beta ]$, $\omega \in \Omega $ and $t\in {\mathbb{R}}$,
\begin{eqnarray}
&&\left\Vert \left( \mathrm{e}^{-it\left( \Delta _{\mathrm{d}}+\lambda
V_{\omega }\right) }-\mathrm{e}^{-it\Delta _{\mathrm{d}}}\right) F_{\alpha
}^{\beta }\left( \Delta _{\mathrm{d}}\right) +\frac{1}{2}\int_{0}^{1}\mathrm{%
e}^{-it(1-\gamma )\Delta _{\mathrm{d}}}it\lambda V_{\omega }\mathrm{e}%
^{-it\gamma \left( \Delta _{\mathrm{d}}+\lambda V_{\omega }\right) }\mathrm{d%
}\gamma \right\Vert _{\mathrm{op}}  \notag \\
&\leq &D\beta \lambda \left\vert t\right\vert \ .  \label{positivity eq3bis}
\end{eqnarray}%
Meanwhile, note that%
\begin{eqnarray}
&&\int_{0}^{1}\mathrm{e}^{-it(1-\gamma )\Delta _{\mathrm{d}}}itV_{\omega }%
\mathrm{e}^{-it\gamma \left( \Delta _{\mathrm{d}}+\lambda V_{\omega }\right)
}\mathrm{d}\gamma  \label{positivity eq3bis+1} \\
&=&itV_{\omega }+\frac{t^{2}}{2}\left( V_{\omega }\Delta _{\mathrm{d}%
}+\Delta _{\mathrm{d}}V_{\omega }+\lambda V_{\omega }^{2}\right) +\mathcal{O}%
(\left\vert t\right\vert ^{3})  \notag
\end{eqnarray}%
uniformly for $\lambda \in \lbrack 0,\lambda _{0}]$ and $\omega \in \Omega $%
. Thus, by combining (\ref{positivity eq30})--(\ref{positivity eq3bis+1}),
for $\lambda \in \lbrack 0,\lambda _{0}]$, we arrive at the equality
\begin{eqnarray}
\Gamma _{1,1}(t) &=&\underset{l\rightarrow \infty }{\lim }\frac{1}{%
\left\vert \Lambda _{l}\right\vert }\sum\limits_{x,y\in \Lambda _{l}}\mathbb{%
E}\left[ \int\nolimits_{0}^{\beta }\widehat{\mathfrak{C}}_{t+i\alpha
}^{(\omega )}(x,x-e_{1},y,y-e_{1})\mathrm{d}\alpha \right]  \notag \\
&&+\mathcal{O}(\beta ^{2}\lambda )+\mathcal{O}(\beta \lambda \left\vert
t\right\vert ^{3})  \label{positivity eq5}
\end{eqnarray}%
for sufficiently small $\beta $ and $\left\vert t\right\vert $, where
\begin{equation*}
\widehat{\mathfrak{C}}_{t+i\alpha }^{(\omega )}(\mathbf{x},\mathbf{y}):=%
\underset{\pi ,\pi ^{\prime }\in S_{2}}{\sum }\varepsilon _{\pi }\varepsilon
_{\pi ^{\prime }}\hat{C}_{t+i\alpha }^{(\omega )}(y^{\pi ^{\prime
}(1)},x^{\pi (1)})\hat{C}_{-t+i(\beta -\alpha )}^{(\omega )}(x^{\pi
(2)},y^{\pi ^{\prime }(2)})
\end{equation*}%
for all $\mathbf{x}:=(x^{(1)},x^{(2)})\in \mathfrak{L}^{2}$ and $\mathbf{y}%
:=(y^{(1)},y^{(2)})\in \mathfrak{L}^{2}$.

We now use that $V_{\omega }$ is an i.i.d. potential satisfying $\mathbb{E}%
[V_{\omega }(x)]=0$ for all $x\in \mathfrak{L}$ to compute that, for any $%
\mathbf{x}:=(x^{(1)},x^{(2)})$ and $\mathbf{y}:=(y^{(1)},y^{(2)})\in
\mathfrak{L}^{2}$, $x^{(1)}\neq x^{(2)}$, $y^{(1)}\neq y^{(2)}$,%
\begin{eqnarray}
&&\mathbb{E}\left[ \int\nolimits_{0}^{\beta }\widehat{\mathfrak{C}}%
_{t+i\alpha }^{(\omega )}(\mathbf{x},\mathbf{y})\mathrm{d}\alpha \right]
-\int\nolimits_{0}^{\beta }\mathfrak{C}_{t+i\alpha }^{(0)}(\mathbf{x},%
\mathbf{y})\mathrm{d}\alpha  \label{important inequality} \\
&=&-\frac{\lambda ^{2}t^{2}}{4}\mathbb{E}\left[ V_{\omega }^{2}\right]
\sum_{\pi ,\pi ^{\prime }\in S_{2}}\varepsilon _{\pi }\varepsilon _{\pi
^{\prime }}\left\{ \left( \int\nolimits_{0}^{\beta }\langle \mathfrak{e}%
_{x^{\pi (1)}},\mathrm{e}^{-it\Delta _{\mathrm{d}}}F_{\alpha }^{\beta
}\left( \Delta _{\mathrm{d}}\right) \mathfrak{e}_{y^{\pi ^{\prime
}(1)}}\rangle \mathrm{d}\alpha \right) \delta _{x^{\pi (2)},y^{\pi ^{\prime
}(2)}}\right.  \notag \\
&&+\left. \left( \int\nolimits_{0}^{\beta }\langle \mathfrak{e}_{y^{\pi
^{\prime }(2)}},\mathrm{e}^{it\Delta _{\mathrm{d}}}F_{\beta -\alpha }^{\beta
}\left( \Delta _{\mathrm{d}}\right) \mathfrak{e}_{x^{\pi (2)}}\rangle
\mathrm{d}\alpha \right) \delta _{y^{\pi ^{\prime }(1)},x^{\pi (1)}}\right\}
+\frac{\beta \lambda ^{2}t^{4}}{16}\mathbf{D}\left( \mathbf{x},\mathbf{y}%
\right) \ ,  \notag
\end{eqnarray}%
where, for any $\mathbf{x}=(x^{(1)},x^{(2)}),\mathbf{y}=(y^{(1)},y^{(2)})\in
\mathfrak{L}^{2}$, $x^{(1)}\neq x^{(2)}$, $y^{(1)}\neq y^{(2)}$,
\begin{eqnarray*}
\mathbf{D}\left( \mathbf{x},\mathbf{y}\right) &:=&\underset{\pi ,\pi
^{\prime }\in S_{2}}{\sum }\varepsilon _{\pi }\varepsilon _{\pi ^{\prime
}}\left\{ \lambda ^{2}\left( \mathbb{E}\left[ V_{\omega }^{2}\right] \right)
^{2}\delta _{y^{\pi ^{\prime }(1)},x^{\pi (1)}}\delta _{x^{\pi (2)},y^{\pi
^{\prime }(2)}}\right. \\
&&\left. +\mathbb{E}\left[ \langle \mathfrak{e}_{x^{\pi (1)}},\left(
V_{\omega }\Delta _{\mathrm{d}}+\Delta _{\mathrm{d}}V_{\omega }\right)
\mathfrak{e}_{y^{\pi ^{\prime }(1)}}\rangle \langle \mathfrak{e}_{y^{\pi
^{\prime }(2)}},\left( V_{\omega }\Delta _{\mathrm{d}}+\Delta _{\mathrm{d}%
}V_{\omega }\right) \mathfrak{e}_{x^{\pi (2)}}\rangle \right] \right\} \ .
\end{eqnarray*}%
Note that, for each $\lambda \in \mathbb{R}_{0}^{+}$ and $t\in \mathbb{R}$, $%
\mathbf{D}\equiv \mathbf{D}^{(\lambda )}$ can be seen as the kernel (w.r.t.
the canonical basis $\{\mathfrak{e}_{x}\otimes \mathfrak{e}_{x^{\prime
}}\}_{x,x^{\prime }\in \mathfrak{L}}$) of a bounded operator on $\ell ^{2}(%
\mathfrak{L})\otimes \ell ^{2}(\mathfrak{L})$ with operator norm uniformly
bounded w.r.t. $\lambda $ on compact sets. Therefore, it is straightforward
to deduce that%
\begin{equation}
\underset{l\rightarrow \infty }{\lim }\frac{1}{\left\vert \Lambda
_{l}\right\vert }\sum\limits_{x,y\in \Lambda _{l}}\mathbf{D}%
(x,x-e_{1},y,y-e_{1})=\mathcal{O}(1)  \label{positivity eq70}
\end{equation}%
uniformly for $\lambda $ in compact sets. For more details on the last
equation, see for instance the proofs of \cite[Lemma 5.3]{OhmII} and \cite[%
Lemma 5.10]{OhmIII}.

Because of Lemma \ref{lemma conductivty4 copy(1)} and (\ref{positivity eq2}%
)--(\ref{positivity eq1}), note that
\begin{eqnarray*}
&&\underset{l\rightarrow \infty }{\lim }\frac{1}{\left\vert \Lambda
_{l}\right\vert }\sum\limits_{x,y\in \Lambda _{l}}\int\nolimits_{0}^{\beta }%
\mathfrak{C}_{t+i\alpha }^{(0)}(x,x-e_{1},y,y-e_{1})\mathrm{d}\alpha \\
&=&\underset{l\rightarrow \infty }{\lim }\frac{1}{\left\vert \Lambda
_{l}\right\vert }\sum\limits_{x,y\in \Lambda _{l}}\int\nolimits_{0}^{\beta }%
\mathfrak{C}_{i\alpha }^{(0)}(x,x-e_{1},y,y-e_{1})\mathrm{d}\alpha
\end{eqnarray*}%
does not depend on $t\in \mathbb{R}$. Using this, for $\lambda \in \lbrack
0,\lambda _{0}]$, we infer from (\ref{positivity eq1}) and (\ref{positivity
eq5})--(\ref{positivity eq70}) the existence of a constant $D\in \mathbb{R}%
^{+}$ such that the paramagnetic conductivity $\sigma _{\mathrm{p}}$ is of
the form%
\begin{equation}
\sigma _{\mathrm{p}}(t)=-D\lambda ^{2}\beta t^{2}+\mathcal{O}(\beta
^{2}\lambda )+\mathcal{O}(\beta \lambda \left\vert t\right\vert ^{3})
\label{positivity 6}
\end{equation}%
for $\lambda \in \lbrack 0,\lambda _{0}]$ and sufficiently small $\beta
,\left\vert t\right\vert $.

Now we choose sufficiently small $\beta _{0},T_{0}>0$ and estimate the
energy increment caused by the time--rescaled potential $\mathbf{A}^{(T)}\in
C_{0}^{\infty }(\mathbb{R}\times {\mathbb{R}}^{d};{\mathbb{R}}%
^{d})\backslash \{0\}$ for $T\in (T_{0}/2,T_{0})$, $\lambda \in (\lambda
_{0}/2,\lambda _{0})$, $\beta \in (0,\beta _{0})$. We assume w.l.o.g. that $%
E_{\mathbf{A}}$ is zero in all but the first component which equals a
function $\mathcal{E}_{t}\in C_{0}^{\infty }(\mathbb{R}^{d};\mathbb{R})$ for
any $t\in \mathbb{R}$. Then, by (\ref{positivity 6}) and Fubini's theorem,
we have%
\begin{eqnarray}
&&\int\nolimits_{\mathbb{R}}\mathrm{d}s_{1}\int\nolimits_{\mathbb{R}}\mathrm{%
d}s_{2}\mathbf{\sigma }_{\mathrm{p}}(s_{1}-s_{2})\int\nolimits_{\mathbb{R}%
^{d}}\mathrm{d}^{d}x\left\langle E_{\mathbf{A}^{(T)}}(s_{2},x),E_{\mathbf{A}%
^{(T)}}(s_{1},x)\right\rangle  \notag \\
&=&-D\lambda ^{2}\beta T^{2}\int\nolimits_{\mathbb{R}^{d}}\mathrm{d}%
^{d}x\int\nolimits_{\mathbb{R}}\mathrm{d}s_{1}\int\nolimits_{\mathbb{R}}%
\mathrm{d}s_{2}(s_{1}-s_{2})^{2}\mathcal{E}_{s_{2}}(x)\mathcal{E}_{s_{1}}(x)
\notag \\
&&+\mathcal{O}(\beta ^{2}\lambda )+\mathcal{O}(\beta \lambda T^{3})\ .
\label{positivity 7}
\end{eqnarray}%
Because $\mathbf{A}\in C_{0}^{\infty }(\mathbb{R}\times {\mathbb{R}}^{d};{%
\mathbb{R}}^{d})\backslash \{0\}$, we infer from (\ref{V bar 0}) that
\begin{equation*}
\int\nolimits_{\mathbb{R}}\mathcal{E}_{s}(x)\mathrm{d}s=0
\end{equation*}%
and, for all $x\in \mathbb{R}^{d}$,
\begin{equation}
-\int\nolimits_{\mathbb{R}}\mathrm{d}s_{1}\int\nolimits_{\mathbb{R}}\mathrm{d%
}s_{2}(s_{1}-s_{2})^{2}\mathcal{E}_{s_{2}}(x)\mathcal{E}_{s_{1}}(x)=2\left(
\int\nolimits_{\mathbb{R}}s\mathcal{E}_{s}(x)\mathrm{d}s\right) ^{2}\ .
\label{positivity 7*}
\end{equation}%
As a consequence, if
\begin{equation*}
\int\nolimits_{\mathbb{R}^{d}}\left( \int\nolimits_{\mathbb{R}}s\mathcal{E}%
_{s}(x)\mathrm{d}s\right) ^{2}\mathrm{d}^{d}x>0\ ,
\end{equation*}%
then (\ref{positivity 7})--(\ref{positivity 7*}) yield the lemma, provided $%
\lambda _{0}T_{0}^{2}\gg \beta _{0},T_{0}^{3}$\ .
\end{proof}

Note that Lemma \ref{main 2 copy(2)} implies that, for any $\lambda \in
\mathbb{R}^{+}$ and sufficiently small $\beta \in \mathbb{R}^{+}$, the
AC--conductivity measure is non--zero, i.e.,
\begin{equation}
\mu _{\mathrm{AC}}\left( \mathbb{R}\backslash \{0\}\right) =\mu _{\mathrm{p}%
}\left( \mathbb{R}\backslash \{0\}\right) >0\ .
\label{non zero cond measure}
\end{equation}%
This property implies the following result:

\begin{lemma}[Non--vanishing AC--conductivity measure -- II]
\label{lemma conductivty4 copy(7)2}\mbox{
}\newline
If (\ref{non zero cond measure}) holds then the set
\begin{equation*}
\mathcal{Z}:=\left\{ \varphi \in \mathcal{S}\left( \mathbb{R};\mathbb{R}%
\right) :\int\nolimits_{\mathbb{R}}\mathrm{d}s_{1}\int\nolimits_{\mathbb{R}}%
\mathrm{d}s_{2}\ \mathbf{\sigma }_{\mathrm{p}}(s_{2}-s_{1})\varphi
(s_{1})\varphi (s_{2})=0\right\}
\end{equation*}%
is meager in the Fr\'{e}chet space $\mathcal{S}\left( \mathbb{R};\mathbb{R}%
\right) $ of Schwartz functions equipped with the usual locally convex
topology.
\end{lemma}

\begin{proof}
By (\ref{non zero cond measure}), there is at least one point $\nu _{0}\in
\mathbb{R}\backslash \{0\}$ such that $\mu _{\mathbf{\Sigma }}\left(
\mathcal{V}\right) \neq 0$ for all open neighborhoods $\mathcal{V}$ of $\nu
_{0}$. To see this, observe that
\begin{equation*}
\mathbb{R}\backslash \{0\}=\underset{n\in \mathbb{N}}{\mathop{\displaystyle
\bigcup }}\left[ \frac{1}{n},n\right] \cup \left[ -n,-\frac{1}{n}\right] \ ,
\end{equation*}%
and thus there is $n\in \mathbb{N}$ such that
\begin{equation*}
\mu _{\mathrm{AC}}\left( \left[ \frac{1}{n},n\right] \cup \left[ -n,-\frac{1%
}{n}\right] \right) >0\ .
\end{equation*}%
Then, by compactness, there is $\nu _{0}\in \left[ \frac{1}{n},n\right] \cup %
\left[ -n,-\frac{1}{n}\right] $ such that
\begin{equation*}
\mu _{\mathrm{AC}}\left( \mathcal{V}\cap \left( \left[ \frac{1}{n},n\right]
\cup \left[ -n,-\frac{1}{n}\right] \right) \right) \neq 0
\end{equation*}%
for all open neighborhoods $\mathcal{V}$ of $\nu _{0}$.

Take now any non--zero function $\varphi \in C_{0}^{\infty }\left( \mathbb{R}%
;\mathbb{R}\right) \subset \mathcal{S}\left( \mathbb{R};\mathbb{R}\right) $.
By the Palay--Wiener theorem, its Fourier transform $\hat{\varphi}:\mathbb{%
R\rightarrow C}$ uniquely extends to an entire function $\mathbb{%
C\rightarrow C}$, again denoted by $\hat{\varphi}$. Hence, the set of zeros
of $\hat{\varphi}$ has no accumulation points.

If $\hat{\varphi}\left( \nu _{0}\right) \neq 0$ then, by continuity of $\hat{%
\varphi}$,
\begin{equation}
\int\nolimits_{\mathbb{R}}\int\nolimits_{\mathbb{R}}\mathbf{\sigma }_{%
\mathrm{p}}(s_{1}-s_{2})\varphi (s_{2})\varphi (s_{1})\mathrm{d}s_{2}\mathrm{%
d}s_{1}=\int\nolimits_{\mathbb{R}\backslash \{0\}}\left\vert \hat{\varphi}%
(\nu )\right\vert ^{2}\mu _{\mathrm{AC}}\left( \mathrm{d}\nu \right) >0\ .
\label{first eq}
\end{equation}%
If $\hat{\varphi}\left( \nu _{0}\right) =0$ then, for all $\alpha \in (0,1)$%
, we define the rescaled function $\hat{\varphi}_{\alpha }\left( \nu \right)
$ by $\hat{\varphi}\left( \alpha \nu \right) $, which is the Fourier
transform of $\alpha ^{-1}\varphi \left( \alpha ^{-1}x\right) $. For
sufficiently small $\varepsilon \in \mathbb{R}^{+}$ and all $\alpha \in
(1-\varepsilon ,1)$,
\begin{equation*}
\int\nolimits_{\mathbb{R}\backslash \{0\}}\left\vert \hat{\varphi}_{\alpha
}\left( \nu \right) \right\vert ^{2}\mu _{\mathrm{AC}}\left( \mathrm{d}\nu
\right) >0\ ,
\end{equation*}%
because the set of zeros of $\hat{\varphi}$ has no accumulation points. On
the other hand, $\alpha ^{-1}\varphi \left( \alpha ^{-1}x\right) $ converges
in $\mathcal{S}\left( \mathbb{R};\mathbb{R}\right) $ to $\varphi \left(
x\right) $, as $\alpha \rightarrow 1$. Thus, the complement of $\mathcal{Z}$%
\ is dense in $\mathcal{S}\left( \mathbb{R};\mathbb{R}\right) $, by density
of the set $C_{0}^{\infty }\left( \mathbb{R};\mathbb{R}\right) $ in $%
\mathcal{S}\left( \mathbb{R};\mathbb{R}\right) $. Since $\mu _{\mathrm{AC}%
}:=\mu _{\mathrm{p}}|_{\mathbb{R}\backslash \{0\}}$ with $\mu _{\mathrm{p}}(%
\mathbb{R})<\infty $ (Theorem \ref{lemma sigma pos type copy(4)-macro}),
note that the map
\begin{equation*}
\hat{\varphi}\mapsto \int\nolimits_{\mathbb{R}\backslash \{0\}}\left\vert
\hat{\varphi}(\nu )\right\vert ^{2}\mu _{\mathrm{AC}}\left( \mathrm{d}\nu
\right)
\end{equation*}%
is continuous on $\mathcal{S}\left( \mathbb{R};\mathbb{R}\right) $. Because
the Fourier transform is a homeomorphism of $\mathcal{S}\left( \mathbb{R};%
\mathbb{R}\right) $, by the first equation in (\ref{first eq}), the map
\begin{equation*}
\varphi \mapsto \int\nolimits_{\mathbb{R}}\int\nolimits_{\mathbb{R}}\mathbf{%
\sigma }_{\mathrm{p}}(s_{1}-s_{2})\varphi (s_{2})\varphi (s_{1})\mathrm{d}%
s_{2}\mathrm{d}s_{1}
\end{equation*}%
is also continuous on $\mathcal{S}\left( \mathbb{R};\mathbb{R}\right) $ and
the complement of $\mathcal{Z}$ is hence an open set.
\end{proof}

\bigskip

\noindent \textit{Acknowledgments:} We would like to thank Volker Bach,
Horia Cornean, Abel Klein and Peter M\"{u}ller for relevant references and
interesting discussions as well as important hints. JBB and WdSP are also
very grateful to the organizers of the Hausdorff Trimester Program entitled
\textquotedblleft \textit{Mathematical challenges of materials science and
condensed matter physics}\textquotedblright\ for the opportunity to work
together on this project at the Hausdorff Research Institute for Mathematics
in Bonn. This work has also been supported by the grant MTM2010-16843 and the BCAM Severo Ochoa accreditation SEV-2013-0323
(MINECO) as well as the FAPESP grant 2013/13215--5 and the
Basque Government through the grant IT641-13 and the BERC 2014-2017 program.


\begin{thebibliography}{99}
\bibitem{OhmI} \textsc{J.-B. Bru, W. de Siqueira Pedra and C. Hertling}, Heat Production of Non--Interacting Fermions Subjected to
Electric Fields, \textit{Comm. Pure Appl. Math.} \textbf{68}(6) (2015), 964--1013.

\bibitem{OhmII} \textsc{J.-B. Bru, W. de Siqueira Pedra and C. Hertling%
}, Microscopic Conductivity of Lattice Fermions at Equilibrium -- Part I:
Non--Interacting Particles, \textit{J. Math. Phys.} \textbf{56} (2015)
051901-1--51.

\bibitem{OhmIII} \textsc{J.-B. Bru, W. de Siqueira Pedra and C.
Hertling}, AC--Conductivity Measure from Heat Production of Free Fermions in
Disordered Media. \textit{Archive for Rational Mechanics and
Analysis} \textbf{220} (2016) 445–-504.

\bibitem{Annale} \textsc{A. Klein, O. Lenoble, and P. M\"{u}ller}, On Mott's
formula for the ac-conductivity in the Anderson model, \textit{Annals of
Mathematics} \textbf{166} (2007) 549--577.

\bibitem{JMP-autre} \textsc{A. Klein and P. M\"{u}ller}, The Conductivity
Measure for the Anderson Model, \textit{Journal of Mathematical Physics,
Analysis, Geometry} \textbf{4} (2008) 128--150.

\bibitem{BratteliRobinson} \textsc{O. Bratteli and D.W. Robinson}, \textit{%
Operator Algebras and Quantum Statistical Mechanics, Vol. II, 2nd ed.}
Springer-Verlag, New York, 1996.

\bibitem{AttalJoyePillet2006a} \textsc{C.-A. Pillet}, Quantum Dynamical
Systems, in \textit{Open Quantum Systems I: The Hamiltonian Approach,} Vol.
1880 of Lecture Notes in Mathematics, editors: S.\ Attal, A.\ Joye, C.-A.
Pillet. Springer--Verlag, 2006, p. 107--182.

\bibitem{GVV1} \textsc{D. Goderis, A. Verbeure and P. Vets}, Noncommutative
central limits. \textit{Probab. Theory Related Fields} \textbf{82} (1989)
527--544.

\bibitem{GVV2} \textsc{D. Goderis, A. Verbeure and P. Vets}, Theory of
quantum fluctuations and the onsager relations, \textit{J. Stat. Phys.}
\textbf{56} (1989) 721--746.

\bibitem{GVV3} \textsc{D. Goderis, A. Verbeure and P. Vets}, \textit{About
the mathematical theory of quantum fluctuations,} In: Mathematical Methods
in Statistical Mechanics. Leuven Notes Math. Theoret. Phys. Ser. A Math.
Phys., \textbf{1}, Leuven, Leuven Univ. Press, 1989, p. 31.

\bibitem{GVV4} \textsc{D. Goderis, A. Verbeure and P. Vets}, \textit{Quantum
central limit and coarse graining. In: Quantum probability and applications,
V,} Vol. 1442 of Lecture Notes in Math. Berlin-Heidelberg-New York,
Springer, 1990, p. 178--193.

\bibitem{GVV5} \textsc{D. Goderis, A. Verbeure and P. Vets}, Dynamics of
fluctuations for quantum lattice systems, \textit{Commun. Math. Phys.}
\textbf{128} (1990) 533--549.

\bibitem{GVV6} \textsc{D. Goderis, A. Verbeure and P. Vets}, About the
exactness of the linear response theory, \textit{Commun. Math. Phys.}
\textbf{136} (1991) 265--283.
\end{thebibliography}
\end{document}